\documentclass[12pt]{article}
\usepackage{natbib,amssymb,times,amsthm,amsmath}
\usepackage[pdftex]{graphicx}
\usepackage{rotating}
\usepackage{multirow}
\usepackage{setspace}
\usepackage{tabularx}
\usepackage{xcolor}

\usepackage[letterpaper, left=1in, top=1in, right=0.75in, bottom=1in, nohead,verbose, ignoremp]{geometry}

\graphicspath{{./}}

\newtheorem{proposition}{Proposition}
\newcommand{\Ex}[1]{\mbox{E}\left[ #1 \right]}

\begin{document}

\title{ \vspace{-2.5cm} Sampling the Bayesian Elastic Net}

\author{Christopher M.~Hans\footnote{Corresponding author: hans@stat.osu.edu}~~and Ningyi Liu\\\vspace{-0.5cm}
{\small Department of Statistics, The Ohio State University, Columbus, OH, 43210, USA}}

\date{\small December 2024}
\maketitle

%\vspace{-1.25cm}
\abstract{The Bayesian elastic net and its variants have become popular approaches 
to regression in many areas of research. The model is characterized by the prior distribution on the 
regression coefficients, the negative log density of which corresponds to the elastic net penalty function. 
While straightforward Markov chain Monte Carlo (MCMC) methods exist for sampling from the posterior
distribution of the regression coefficients given the penalty parameters, full Bayesian inference---where
the MCMC algorithms are expanded to integrate over uncertainty in the penalty parameters---remains a
challenge. Sampling the penalty parameters (and the regression model error's variance parameter under
some forms of the prior) is complicated by the presence of an intractable
integral expression in the normalizing constant for the prior on the 
regression coefficients. Though sampling methods have been proposed that avoid the need to compute the 
normalizing constant, all correctly-specified methods for updating the remaining parameters that have been 
described in the literature involve at least one ``Metropolis-within-Gibbs'' 
update, requiring specification and tuning of proposal distributions. The computational landscape is
complicated by the fact that two different forms of the Bayesian elastic net prior have been introduced in
the literature, and two different representations (with and without data augmentation) of the prior suggest 
different MCMC algorithms for sampling the regression coefficients. We first provide a comprehensive review
of the forms and representations of the Bayesian elastic net prior, discussing all combinations of these
different treatments of the prior together for the first time and introducing one combination of form
and representation that has yet to appear in the literature. We then introduce MCMC algorithms for
full Bayesian inference for all combinations of prior form and representation. The algorithms allow
for direct sampling of all parameters at low computational cost without any ``Metropolis-within-Gibbs''
steps, avoiding potential problems with slow convergence and mixing due to poor choice of proposal
distribution. The key to the new approach is a careful transformation of the parameter space and an
analysis of the resulting full conditional density functions that allows for efficient rejection sampling 
of the transformed parameters.
We make empirical comparisons between our sampling approaches and other existing MCMC methods 
in the literature for a variety of potential data structures.}

\bigskip
\noindent{\textit{Key Words:} elastic net; lasso; 
MCMC; orthant normal distribution; prior distribution; regression; regularization; rejection sampling; shrinkage}

\newpage

\section{Introduction}\label{sec:intro}
The Bayesian elastic net and its variants have become popular approaches to regression in many areas of 
research. \citet{li:10} and \citet{hans:11} introduced the Bayesian elastic net model in the normal linear 
regression setting, $y = \mathbf{1}\alpha + X\beta + \varepsilon$, where $y$ is an $n \times 1$ response vector,
$X$ is an $n \times p$ matrix of regressors, and $\varepsilon \sim \mbox{N}(0, \sigma^2 I_n)$. In this framework,
the Bayesian elastic net is characterized by the prior on the regression coefficients,
\begin{equation}
	\pi_c(\beta \mid \sigma^2, \lambda_1, \lambda_2) \propto
		\exp\left\{ -\frac{1}{2\sigma^2} \left( \lambda_2 \beta^T \beta + \lambda_1 |\beta|_1
		\right)\right\},   \label{eq:csprior}
\end{equation}
where $|\beta|_1 = \sum_{j=1}^p |\beta_j|$ is the $\ell_1$-norm of $\beta$. 
Under the non-informative prior on the intercept parameter, $\pi(\alpha) \propto 1$, the integrated likelihood 
function is
\begin{equation}
	p(y \mid X, \beta, \sigma^2) = \int p(y \mid X, \alpha, \beta, \sigma^2) p(\alpha) d\alpha =
	(2\pi\sigma^2)^{-(n-1)/2}n^{-1/2}e^{-\frac{1}{2\sigma^2}(y^* - X^*\beta)^T(y^* - X^*\beta)},
	\label{eq:intlik}
\end{equation}
where $y^*$ and $X^*$ are the mean-centered response vector and column-mean-centered matrix of
regressors, respectively. The posterior distribution $\pi_c(\beta \mid y, \sigma^2)$ then satisfies
\begin{equation}
	-2\sigma^2 \log \pi_c(\beta \mid y, \sigma^2) = \mbox{const. } + (y^* - X^* \beta)^T (y^* - X^* \beta) + 
		\lambda_2 \beta^T\beta + \lambda_1 |\beta|_1. \label{eq:logpostEN}
\end{equation}
The non-constant component of (\ref{eq:logpostEN}) is the elastic net objective function
\citep{zou:05} with penalty parameters $\lambda_1$ and $\lambda_2$ so that, for any fixed value of $\sigma^2$, 
the posterior mode of $\beta$ corresponds to an elastic net estimate. 
We use the integrated likelihood (\ref{eq:intlik}) throughout and assume that $y$ and the columns of $X$ have 
been mean-centered, dropping the ``$^*$'' in the notation.

The literature on the connection between Bayesian posterior modes and estimators described as solutions
to penalized optimization problems is quite rich. \citet{tibs:96} made the first such connection for 
lasso regression, the Bayesian side of which was more fully developed by \citet{park:08} and 
\citet{hans:09, hans:10}. Bayesian connections to the adaptive lasso \citep{zou:06} have been considered
by \citet{grif:07, grif:11}, \citet{alha:12}, \citet{leng:14}, \citet{alha:18}, \citet{kang:19}, and \citet{wang:19},
among others. \citet{rock:18} introduced a fully-Bayes, adaptive approach to Bayesian variable selection
that combines the lasso penalty function with the ideas that underly ``spike-and-slab'' priors.
\citet{wang:12} introduced a Bayesian formulation of the graphical lasso \citep{mein:06, yuan:07, frie:08}, 
while \citet{kyun:10} and others have studied connections to the group and fused lassos \citep{yuan:06,  tibs:05}.

Bayesian regression models with connections to the elastic net have also received extensive attention in the literature.
After \citet{zou:05} noted that their elastic net estimator could be viewed as the mode of a Bayesian posterior
distribution, \citet{kyun:10}, \citet{li:10} and \citet{hans:11} sought to fully characterize the corresponding Bayesian
model. \citet{li:10} represented the prior as a scale-mixture of normal distributions, discussed inference when 
$\sigma^2$ was unknown, and introduced Bayesian approaches for selecting the
penalty parameters. \citet{hans:11} also described the scale-mixture of normals representation, introduced an 
additional, ``direct'' representation of the prior, considered full Bayesian inference on the penalty parameters, and 
introduced methods for Bayesian elastic net variable selection and model averaging for prediction. More recently, 
\citet{lee:15} identified an error in \citet{kyun:10}'s representation of the
elastic net prior. \citet{roy:17} corrected the error and also studied optimal selection of 
the penalty parameters. \citet{wang:23} introduced an MCMC algorithm for full Bayesian elastic net inference
that was designed to avoid the need to approximate any integrals in any of the sampling steps.

Two main forms of the Bayesian elastic net prior distribution are common in the literature.
\citet{li:10} and \citet{hans:11} considered the prior as parameterized in (\ref{eq:csprior}), where
the two components of the penalty function, $\lambda_2 \beta^T\beta$ and $\lambda_1 |\beta_1|$,
are both scaled by $2\sigma^2$. We refer to this form of the prior as the 
``commonly-scaled'' parameterization, and we subscript prior and posterior densities under the common scaling 
with ``$c$'' for clarity, e.g., $\pi_c(\beta \mid \sigma^2, \lambda_1, \lambda_2)$ in (\ref{eq:csprior}).
\citet{kyun:10} and \citet{roy:17} scale 
the penalty terms differentially:
\begin{eqnarray}
	\pi_d(\beta \mid \sigma^2, \lambda_1 \lambda_2) \propto \exp\left\{-\frac{\lambda_2}{2\sigma^2}\beta^T\beta
		- \frac{\lambda_1}{\sigma}|\beta|_1\right\}. \label{eq:altscale}
\end{eqnarray}
This version of the prior has the useful property that $\lambda_1$ and $\lambda_2$ do not depend on the units
of the response variable: coupled with a scale-invariant prior on $\sigma$, an analyst would not need to adjust
the fixed values of (or priors for) $\lambda_1$ and $\lambda_2$ in order to obtain the same posterior if the 
response variable was rescaled linearly. We refer to this form of the prior as the ``differentially-scaled'' parameterization
and subscript corresponding prior and posterior densities with ``$d$'' to emphasize the
differential scaling. We omit the subscript when the distinction between the two forms of the prior  is not relevant. 
For fixed values of $\sigma^2$, $\lambda_1$, and $\lambda_2$, the differentially-scaled
prior is simply a reparameterization of the commonly-scaled prior: given the same values of $\sigma^2$
and $\lambda_2$, one can obtain the same posterior distribution under both priors through choice 
of the $\lambda_1$ parameter specific to each prior. It is useful, though, to consider both versions of the
prior separately because the interpretation of $\lambda_1$ is specific to the form of the scaling.

While the papers that study the Bayesian elastic net take different perspectives and have varying objectives,
they have in common the theme that full Bayesian inference under the elastic net prior presents 
computational challenges. The $|\beta|_1$ term in the prior makes direct integration of the posterior
challenging, and so one of the keys to Bayesian elastic net regression modeling is the ability to easily sample
from the posterior. Several MCMC algorithms have been 
proposed to obtain samples from the conditional posterior of $\beta$ given $\sigma^2$, $\lambda_1$ and 
$\lambda_2$. One approach---similar to the one used by \citet{park:08} for the Bayesian lasso---is to 
demarginalize the prior on $\beta$ by introducing latent variables, $\tau^2$, that can be exploited to conduct a 
data augmentation Gibbs sampler. \citet{li:10} and \citet{hans:11} extend this idea to the elastic net penalty 
function and introduce corresponding data augmentation Gibbs samplers under the common scale prior, and 
\citet{roy:17} and \citet{wang:23} consider the data augmentation approach under the differentially-scaled prior.
As an alternative to data augmentation Gibbs sampling, \citet{hans:11} describes an alternative Gibbs sampler
for the Bayesian elastic net that updates each $\beta_j$ one at a time, conditionally on the others, without requiring the 
inclusion of latent variables in the sampling scheme.

The more difficult challenges to computation become apparent when we assign prior distributions to 
$\sigma^2$, $\lambda_1$, and $\lambda_2$ and wish to make inference based on the joint posterior 
distribution. As seen in Section~\ref{sec:review}, the normalizing constant for the Bayesian elastic net prior
(\ref{eq:csprior}) contains the term $\Phi(-\lambda_1/(2\sigma\sqrt{\lambda_2}))^{-p}$, where
$\Phi(\cdot)$ is the standard normal cumulative distribution function (cdf). The same term appears in the
joint prior density for $\beta$ and $\tau^2$ under the data augmentation representation of the prior.
Sampling from or integrating the joint posterior distribution therefore requires dealing with
$\Phi(-\lambda_1/(2\sigma\sqrt{\lambda_1})^{-p}$, an integral expression with no closed form
solution. Noting that the standard normal cdf can evaluated numerically to relatively high precision when its argument
is not too close to $\pm \infty$, \citet{hans:11}
sampled $\sigma^2$, $\lambda_1$, and $\lambda_2$ via random-walk ``Metropolis-within-Gibbs'' updates
for $\log \sigma^2$, $\log \lambda_1$, and $\log \lambda_2$, with their respective full conditional
distributions as the target distributions. While effective, this approach requires specifying step-size
parameters for the random walks, poor choices of which can lead to slow convergence and mixing and
the need to iteratively tune and re-run the MCMC algorithm. \citet{li:10} avoided the issue of sampling
the penalty parameters by devising data-adaptive methods for selecting values for them, and attempted
to sample $\sigma^2$ directly from its full conditional distribution via rejection sampling. Unfortunately,
as shown in Appendix~\ref{app:sig2rejsamp}, their rejection sampling algorithm contains an error and does not
produce samples from the desired target distribution.

The computational situation is slightly improved when the differentially-scaled form of the prior
(\ref{eq:altscale}) is used. As shown in Section~\ref{sec:review}, the awkward term in the normalizing 
constant is then $\Phi(-\lambda_1/\sqrt{\lambda_2})^{-p}$, which no longer depends on $\sigma^2$. 
Posterior sampling of $\sigma^2$ under this form of the prior is straightforward, as demonstrated
by \citet{roy:17}. Full Bayesian inference under priors on $\lambda_1$ and $\lambda_2$, however, still 
must involve methods for handling the 
analytically intractable integral expression. Motivated by the desire to avoid numerical computation of 
$\Phi(\cdot)$, \citet{wang:23} devised a clever exchange algorithm 
\citep{murr:06} that introduces $p$ additional latent variables in such a way as to remove
the term involving $\Phi(\cdot)$ from the joint posterior of the augmented parameter space.
Despite avoiding computation of $\Phi(\cdot)$, the algorithm still requires one parameter
be updated via the Metropolis--Hastings algorithm using a random-walk proposal,
necessitating the selection of a step-size parameter for the random walk.

All of the correctly-specified MCMC approaches described above for fully-Bayes inference under the 
Bayesian elastic net use at least one Metropolis step. From the point of view of a practitioner, it would be
better if the associated selection of random-walk step-size parameters could be avoided entirely.
The computational landscape is further complicated by the fact that the two different forms
(common and differential scaling) and two different representations (with and without data augmentation)
of the prior suggest different approaches to MCMC for full Bayesian inference for Bayesian elastic net 
regression. In Section~\ref{sec:review}, we provide a comprehensive review of the forms and representations
of the Bayesian elastic net prior and the existing approaches to computation. We discuss together for the first 
time all combinations of the different treatments 
of form and representation of the prior, and we introduce one combination of form and treatment that
has yet to be discussed in the literature. We use this review to highlight the computational difficulties
associated with full Bayesian inference. To solve the computational problems, we introduce in
Section~\ref{sec:rejsamp} new MCMC algorithms for full Bayesian inference under all combinations
of prior form and representation. The algorithms allow for direct sampling of all parameters at low
computational cost without using any ``Metropolis-within-Gibbs'' steps, avoiding potential problems
with slow convergence and mixing due to poor choice of proposal distribution. The key to the new
approach is a careful transformation of the parameter space and an analysis of the resulting
full conditional density functions that allows for efficient rejection sampling of the transformed
parameters. We make empirical comparisons in Section~\ref{sec:sim} between our new sampling 
approaches and other existing MCMC methods in the literature for a variety of potential data
structures.

\section{Existing Approaches to Model Specification and Posterior Computation}\label{sec:review}

The Bayesian elastic net prior distribution can be represented directly (without data augmentation) or hierarchically
(with data augmentation). We refer to these two representations of the prior as the ``direct'' and ``DA'' representations,
respectively. Both forms of the prior (common and differential scaling) have a direct and DA representation.
Table~\ref{tab:formrep} indicates where these combinations of form and representation originally appeared in the 
literature. In this section we review the four combinations of representation and form, provide the corresponding posterior
distributions, and discuss existing approaches to posterior computation, highlighting computational challenges.

\begin{table}[h]
\centering
\begin{tabular}{c|c|c|}
%\multicolumn{1}{c}{} & \multicolumn{1}{c}{Direct} & \multicolumn{1}{c}{Data Augmentation (DA)} \\
%\multicolumn{1}{c}{Representation:} 
\multicolumn{1}{c}{} & \multicolumn{1}{c}{Direct} & \multicolumn{1}{c}{Data Augmentation (DA)} \\ 
	%\multicolumn{1}{c}{Representation} & \multicolumn{1}{c}{Representation} \\
\cline{2-3}
 Common Scaling & \citet{hans:11} & \citet{li:10}, \citet{hans:11} \\ 
 \cline{2-3}
 Differential Scaling & * & \citet{kyun:10}$^{**}$, \citet{roy:17} \\
 \cline{2-3}
\end{tabular}
\caption{Citations for original descriptions of the four combinations of form (rows) and representation (columns)
of the Bayesian elastic net prior. The direct representation of the differentially-scaled prior (*) is introduced
in Section~\ref{sec:directprior}. \citet{kyun:10}'s description of the DA representation of the differentially-scaled 
prior (**) contained an error which was corrected by \citet{roy:17}. \label{tab:formrep}}
\end{table}

\subsection{Direct Representation of the Prior}\label{sec:directprior}
\citet{hans:11} introduced the direct characterization of the commonly-scaled elastic net prior (\ref{eq:csprior}).
Handling the term $|\beta|_1 = \sum_{j=1}^p |\beta_j|$ by treating separately the cases $\beta_j < 0$
and $\beta_j \geq 0$, the independent priors on each $\beta_j$ can be expressed as
two separate, symmetric, truncated normal distributions that are weighted to have matching density at the origin.
Combining the cases together, 
\begin{align}
	\pi_c(\beta \mid \sigma^2, \lambda_1, \lambda_2) &= 
		\prod_{j=1}^p \left\{ \frac{1}{2} \cdot \mbox{N}^{-}\left(\beta_j \left| \frac{\lambda_1}{2\lambda_2}\right. ,
			\frac{\sigma^2}{\lambda_2}\right) + 
			\frac{1}{2} \cdot \mbox{N}^{+}\left(\beta_j \left| -\frac{\lambda_1}{2\lambda_2}\right. ,
			\frac{\sigma^2}{\lambda_2}\right) \right\}  \nonumber \\
		&\equiv  \prod_{j=1}^p \left\{ \frac{1}{2} \cdot \frac{\mbox{N}\left(\beta_j \mid \frac{\lambda_1}{2\lambda_2}, 
			\frac{\sigma^2}{\lambda_2}\right)}{\Phi\left(\frac{-\lambda_1}{2\sigma\sqrt{\lambda_2}} \right)}
			\mathbf{1}(\beta_j < 0) \; + \right. \nonumber \\
		&  \hspace{1.5in} \left. \frac{1}{2} \cdot \frac{\mbox{N}\left(\beta_j \mid -\frac{\lambda_1}{2\lambda_2}, 
			\frac{\sigma^2}{\lambda_2}\right)}{\Phi\left(\frac{-\lambda_1}{2\sigma\sqrt{\lambda_2}} \right)}
			\mathbf{1}(\beta_j \geq 0) 
			\right\} \label{eq:truncproddeets}.
\end{align}
The notation $\mbox{N}^-(x \mid m, s^2)$ and 
$\mbox{N}^+(x \mid m, s^2)$ denotes the normalized density functions for, respectively, negatively and non-negatively 
truncated univariate normal distributions, where $m$ and $s^2$ are, respectively, the mean and variance of an 
underlying, non-truncated, normal random variable. The corresponding density functions are:
\[
	\mbox{N}^-\left( x \mid m, s^2\right) \equiv \frac{\mbox{N}\left( x \mid m, s^2\right)}{\Phi(-m/s)}\mathbf{1}(x<0),
	\;\;\;\;
	\mbox{N}^-\left( x \mid m, s^2\right) \equiv \frac{\mbox{N}\left( x \mid m, s^2\right)}{1 - \Phi(-m/s)}\mathbf{1}(x\geq 0),
\]
where $\mbox{N}(x \mid m, s^2) \equiv (2\pi s^2)^{-1/2} \exp\{-(x-m)^2/(2s^2)\} $ is the probability density function for a 
normal distribution with mean $m$ and
variance $s^2$, and $\Phi(x) \equiv \int_{-\infty}^x (2\pi)^{-1/2}e^{-u^2/2}du$ is the standard normal cdf.
%cumulative distribution function.

Isolating some of the constant terms in (\ref{eq:truncproddeets}) and combining the univariate densities into a multivariate 
density, the prior can also be written as
\begin{equation}
	\pi_c(\beta \mid \sigma^2, \lambda_1, \lambda_2)
		= 2^{-p} \Phi\left(\frac{-\lambda_1}{2\sigma\sqrt{\lambda_2}}\right)^{-p} 
			\sum_{z \in \mathcal{Z}} \mbox{N}\left(\beta \mid -\frac{\lambda_1}{2\lambda_2} z,
			\frac{\sigma^2}{\lambda_2}I_p\right) \mathbf{1}(\beta \in \mathcal{O}_z),
			%\nonumber
			\label{eq:onform}
\end{equation}
where $\mbox{N}(x \mid m, S)$ is the
density function for a $p$-dimensional multivariate normal distribution with mean vector $m$ and covariance
matrix $S$. The sum is taken over all $2^p$ possible $p$-vectors $z$ having elements $z_j \in \{-1,1\}$,
with $\mathcal{Z}$ being the set of all such vectors. The notation $\mathcal{O}_z$ refers to the
orthant of $\mathbb{R}^p$ where each coordinate is restricted by $z$ to be negative ($z_j = -1$) or 
non-negative ($z_j = 1$). The elastic net prior distribution can therefore be thought of
as a collection of truncated multivariate normal distributions defined separately on the $2^p$ orthants
of $\mathbb{R}^p$. The location vector for the normal distribution in each orthant depends on the orthant,
but the orientations of the normal distributions are the same for all orthants. The specific values of the 
location and orientation parameters ensure that the prior density is continuous, but not differentiable, along the 
coordinate axes. \citet{hans:11} calls this an ``orthant normal distribution'' and writes the density as
\begin{eqnarray}
	\pi_c(\beta \mid \sigma^2, \lambda_1, \lambda_2)
		&=& \sum_{z\in \mathcal{Z}} 2^{-p} \mbox{N}^{[z]}\left( \beta \mid -\frac{\lambda_1}{2\lambda_2} z,
			\frac{\sigma^2}{\lambda_2}I_p\right) \nonumber \\ %\label{eq:onpriorcompact} \\
		&\equiv& \sum_{z\in \mathcal{Z}} 2^{-p} \frac{\mbox{N}\left( \beta \mid -\frac{\lambda_1}{2\lambda_2} z,
			\frac{\sigma^2}{\lambda_2}I_p\right)}{\mbox{P}\left(z, \frac{-\lambda_1}{2\lambda_2}z,
			\frac{\sigma^2}{\lambda_2}I_p\right)}\mathbf{1}(\beta \in \mathcal{O}_z), \nonumber
\end{eqnarray}
where $\mbox{N}^{[z]}(\cdot \mid \cdot, \cdot)$ denotes the density function for a multivariate normal distribution 
truncated to orthant $\mathcal{O}_z$, and $\mbox{P}(z, \cdot, \cdot)$ is the probability assigned to that orthant
by the underlying normal distribution. For the commonly-scaled elastic net prior, the orthant probabilities are all equal 
and generate the $\Phi(-\lambda_1/(2\sigma\sqrt{\lambda_2}))^{-p}$ term in (\ref{eq:onform}):
\begin{eqnarray*}
	\mbox{P}\left(z, \frac{-\lambda_1}{2\lambda_2}z, \frac{\sigma^2}{\lambda_2}I_p\right) &=&
		\int_{\mathcal{O}_z} \mbox{N}\left(\beta \mid \frac{-\lambda_1}{2\lambda_2}z, 
			\frac{\sigma^2}{\lambda_2}I_p\right) d\beta \\
		&=& \left[ \prod_{j \; : \; z_j = -1} \int_{-\infty}^0 \mbox{N}\left(\beta_j \mid \frac{\lambda_1}{2\lambda_2},
			\frac{\sigma^2}{\lambda_2}\right) d\beta_j \right] \times
			\left[ \prod_{j \; : \; z_j = 1} \int_{0}^\infty \mbox{N}\left(\beta_j \mid \frac{-\lambda_1}{2\lambda_2},
			\frac{\sigma^2}{\lambda_2}\right) d\beta_j \right] \\
		&=& \left[ \prod_{j \; : \; z_j = -1} \Phi\left( \frac{-\lambda_1}{2\sigma\sqrt{\lambda_2}}\right) \right] \times
			\left[ \prod_{j \; : \; z_j = 1} \left(1- \Phi\left( \frac{\lambda_1}{2\sigma\sqrt{\lambda_2}}\right) \right)
			\right] \\
		&=& \Phi\left( \frac{-\lambda_1}{2\sigma\sqrt{\lambda_2}}\right)^p.
\end{eqnarray*}
Examples of this density function when $p=1$ and $p=2$ are shown in Figure~\ref{fig:enprior}.
We refer to this representation of the prior as the ``direct'' representation under the common scaling.
It is sometimes convenient to work with the prior density by properly normalizing the expression in 
(\ref{eq:csprior}), retaining the $|\beta|_1$ term and avoiding the summation over the orthants:
\begin{equation}
	\pi_c(\beta \mid \sigma^2, \lambda_1, \lambda_2) = 
		2^{-p} (2\pi)^{-p/2}  (\sigma^2/\lambda_2)^{-p/2}
		e^{-\frac{p\lambda_1^2}{8 \sigma^2 \lambda_2}}
		 \Phi\left(-\frac{\lambda_1}{2\sigma\sqrt{\lambda_2}}\right)^{-p}
		\exp\left\{-\frac{\lambda_2}{2\sigma^2} \beta^T \beta - \frac{\lambda_1}{2\sigma^2}|\beta|_1 \right\}.	
		\label{eq:cscalebone}
\end{equation}

The direct representation of the differentially-scaled prior (\ref{eq:altscale}) has not been 
explicitly described in the literature, but it is easy to show that
\begin{eqnarray}
	\pi_d(\beta \mid \sigma^2, \lambda_1 \lambda_2) &=& \prod_{j=1}^p \left\{
		\frac{1}{2}\cdot \mbox{N}^{-}\left(\beta_j \left| \frac{\sigma \lambda_1}{\lambda_2},
		\frac{\sigma^2}{\lambda_2}\right. \right) + \frac{1}{2}\cdot \mbox{N}^{+}\left(\beta_j 
		\left| -\frac{\sigma \lambda_1}{\lambda_2}, \frac{\sigma^2}{\lambda_2} \right.
		\right)\right\} \nonumber \\
	&=& 2^{-p} \Phi\left(-\frac{\lambda_1}{\sqrt{\lambda_2}}\right)^{-p} \sum_{z \in \mathcal{Z}}
		\mbox{N}\left( \beta \left| -\frac{\sigma \lambda_1}{\lambda_2}z, \frac{\sigma^2}{\lambda_2}
		I_p\right.\right) \mathbf{1}(\beta \in \mathcal{O}_z). \label{eq:altdirect}
\end{eqnarray}
The term in the normalizing constant involving $\Phi(\cdot)$ does not
depend on $\sigma^2$ under this scaling of the prior. As in (\ref{eq:cscalebone}), we
can express the properly normalized, differentially-scaled prior in terms of $|\beta|_1$ as:
\begin{equation}
	\pi_d(\beta \mid \sigma^2, \lambda_1, \lambda_2) = 
		2^{-p} (2\pi)^{-p/2}  (\sigma^2/\lambda_2)^{-p/2}
		e^{-\frac{p\lambda_1^2}{2\lambda_2}}
		 \Phi\left(-\frac{\lambda_1}{\sqrt{\lambda_2}}\right)^{-p}
		\exp\left\{-\frac{\lambda_2}{2\sigma^2} \beta^T \beta - \frac{\lambda_1}{\sigma}|\beta|_1 \right\}.	
		\label{eq:dscalebone}
\end{equation}

\begin{figure}[t]
 \begin{center}
 \includegraphics[scale=0.6]{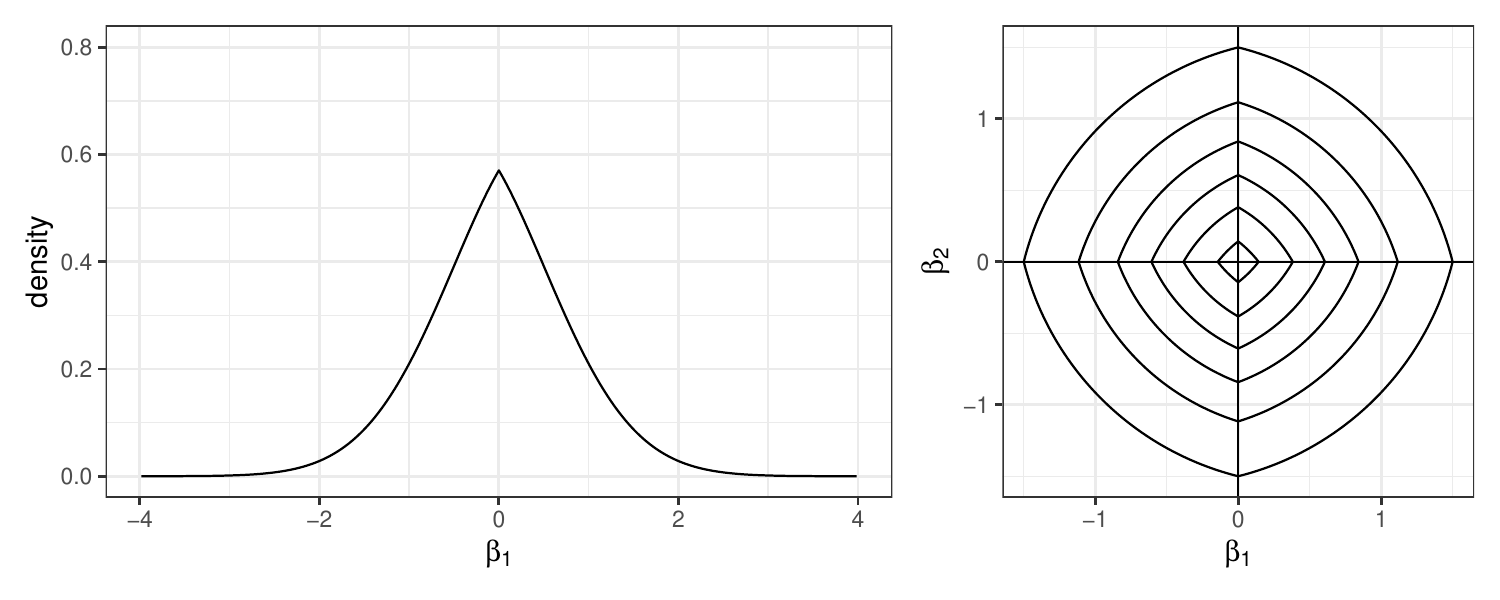}
 \caption{Elastic net prior density $\pi_c(\beta \mid \sigma^2 = 1, \lambda_1 = 1, \lambda_2 = 1)$.
		Left panel: prior density for $\beta_1$ when $p=1$; right panel: contours of
		joint prior density for $\beta_1$ and $\beta_2$ when $p=2$.
 	        \label{fig:enprior}}
 \end{center}
\end{figure}

\subsection{Data Augmentation Representation of the Prior}\label{sec:daprior}

\citet{li:10} and \citet{hans:11} provide an alternate characterization of the commonly-scaled elastic net 
prior distribution by representing it via demarginalization as a scale-mixture of normal distributions 
\citep{andr:74, west:87}. Introducing latent scale variables, $\tau_j^2$, \citet{hans:11} defines the hierarchical 
model
\begin{eqnarray*}
	\beta_j \mid \sigma^2, \lambda_1, \lambda_2, \tau^2_j &\stackrel{\mathrm{ind}}{\sim}&
		\mbox{N}\left(0, \frac{\sigma^2}{\lambda_2}(1-\tau^2_j)\right), \\
	\tau^2_j \mid \sigma^2, \lambda_1, \lambda_2 &\stackrel{\mathrm{iid}}{\sim}&
		\mbox{Inv-Gamma}_{(0,1)}\left( \frac{1}{2}, \frac{1}{2}\left(\frac{\lambda_1}{2\sigma\sqrt{\lambda_2}}\right)^2
		\right),
\end{eqnarray*}
for $j = 1, \ldots, p$, where $\mbox{Inv-Gamma}_{(0,1)}(\cdot, \cdot)$ is an inverse gamma distribution truncated
to the interval $(0,1)$ with density function
\begin{equation}
	\pi_c(\tau^2_j \mid \sigma^2, \lambda_1, \lambda_2) = \frac{1}{2\sqrt{2\pi}}
		\Phi\left(\frac{-\lambda_1}{2\sigma\sqrt{\lambda_2}}\right)^{-1}
		\left(\frac{\lambda_1}{2\sigma\sqrt{\lambda_2}}\right) (\tau_j^2)^{-3/2}
		e^{-\frac{\lambda_1^2}{8\sigma^2 \lambda_2} \tau_j^{-2}}, \;\;
		 0 < \tau_j^2 < 1.
		 \label{eq:directda}
\end{equation}
As shown in \citet{hans:11}, the marginal density of $\beta_j$ under this model is (\ref{eq:cscalebone}), the 
commonly-scaled elastic net prior. \citet{li:10} provide a similar result, but parameterize the scale-mixture 
slightly differently:
\begin{eqnarray*}
	\beta_j \mid \sigma^2, \lambda_1, \lambda_2, \tau^2_j &\stackrel{\mathrm{ind}}{\sim}&
		\mbox{N}\left(0, \frac{\sigma^2}{\lambda_2}\frac{\tau^2_j-1}{\tau^2_j}\right), \\
	\tau^2_j \mid \sigma^2, \lambda_1, \lambda_2 &\stackrel{\mathrm{iid}}{\sim}&
		\mbox{Gamma}_{(1,\infty)}\left( \frac{1}{2}, \frac{1}{2}\left(\frac{\lambda_1}
		{2\sigma\sqrt{\lambda_2}}\right)^2 \right),
\end{eqnarray*}
where the gamma distribution with rate parameter $\lambda_1^2/(8\sigma^2\lambda_2)$ is
truncated to the interval $(1, \infty)$. \citet{li:10} express the normalizing constant in the density 
for the truncated gamma distribution in terms of the upper incomplete gamma function 
\citep{dido:86}, $\Gamma_U(\alpha,x) = \int_x^\infty t^{\alpha-1}e^{-t}dt$. The relevant term
in the normalizing constant of this density is $\Gamma_U\left( \frac{1}{2}, \frac{1}{2}
\left(\frac{\lambda_1}{2\sigma\sqrt{\lambda_2}} \right)^2\right)$, which is equivalent to 
$2\sqrt{\pi}\Phi\left(\frac{-\lambda_1}{2\sigma\sqrt{\lambda_2}}\right)$ due to the identity 
$2\sqrt{\pi}\Phi(-x) = \Gamma_U\left(\frac{1}{2}, \frac{x^2}{2}\right)$ for $x \geq 0$. We refer
to any version of this representation of the prior as the data augmentation (``DA'') or
scale-mixture-of-normals (``SMN'') representation.

\citet{roy:17} introduced the correct DA representation of the differentially-scaled prior (\ref{eq:altscale})
via the hierarchical representation
\begin{eqnarray}
	\beta_j \mid \sigma^2, \lambda_1, \lambda_2, \tau^2_j &\stackrel{\mathrm{ind}}{\sim}& 
		\mbox{N}\left(0, \frac{\sigma^2}{\lambda_2}\left(\frac{\lambda_2 \tau^2_j}{1 + \lambda_2 \tau^2_j}\right)
		\right), \label{eq:correctDA1} \\
	\tau^2_j \mid \lambda_1, \lambda_2 &\stackrel{\mathrm{iid}}{\sim}& \mbox{UH}\left(1, \frac{1}{2},
		\frac{\lambda_1^2}{2}, \lambda_2\right), \label{eq:correctDA2}
\end{eqnarray}
where the $\mbox{UH}(p,r,s,\lambda)$ distribution is a limit of the compound confluent hypergeometric (CCH)
distribution \citep{gord:98}. The density function for this distribution is
\begin{eqnarray}
	\pi_d(\tau^2_j \mid \lambda_1, \lambda_2) = \frac{1}{2\sqrt{2\pi}} \Phi\left(-\frac{\lambda_1}{\sqrt{\lambda_2}}\right)^{-1}
		\lambda_1 \lambda_2^{1/2} e^{-\frac{\lambda_1^2}{2\lambda_2}}
		 (1 + \lambda_2 \tau^2_j)^{-1/2} e^{-\frac{\lambda_1^2
		\tau^2_j}{2}}, \;\;\; \tau_j^2 > 0.
		\label{eq:altda}
\end{eqnarray}
Table~\ref{tab:review} summarizes the different prior 
scalings and data augmentation representations that have appeared in the literature.

%\begin{sidewaystable}
\begin{table}[h]
\centering
\scalebox{0.65}{
\begin{tabular}{rcccc}
	& Prior scaling & SMN variance & Mixing distribution & Mixing distribution  \\
	& $\pi(\beta \mid \sigma^2, \lambda_1, \lambda_2) \propto $ 
		& $\beta_j \mid \sigma^2, \lambda_2, \tau_j^2 \sim \mbox{N}(0, \cdot)$ & $\pi(\tau^2_j \mid \sigma^2, \lambda_1, \lambda_2) \propto$ & 
		family
	\\ \hline\hline \\
	%Li and Lin (2010) 
	\citet{li:10} & $\exp\left\{-\frac{\lambda_2}{2\sigma^2}\beta^T\beta - \frac{\lambda_1}{2\sigma^2}|\beta|_1\right\}$ &
		$\frac{\sigma^2}{\lambda_2}\left(\frac{\tau_j^2-1}{\tau_j^2}\right)$ & 
%		$\frac{\sigma^2}{\lambda_2}\left(1 - \tau_j^{-2}\right)$ & 
		$\left(\tau_j^2\right)^{-\frac{1}{2}} e^{-\tau_j^2 \frac{\lambda_1^2}{8\sigma^2\lambda_2}}$, $\tau^2_j \geq 1,$ &
		$\mbox{Gamma}_{(1,\infty)}\left(\frac{1}{2}, \frac{\lambda_1^2}
		{8\sigma^2\lambda_2}\right)$  \\ \\
	\hline \\
	%Hans (2011) 
	\citet{hans:11} & $\exp\left\{-\frac{\lambda_2}{2\sigma^2}\beta^T\beta - \frac{\lambda_1}{2\sigma^2}|\beta|_1\right\}$
	& $\frac{\sigma^2}{\lambda_2}\left(1 - \tau_j^2\right)$ & $\left(\tau_j^2\right)^{-\frac{3}{2}} e^{-\tau_j^{-2} \frac{\lambda_1^2}
		{8\sigma^2\lambda_2}}$, $0 \leq \tau_j^2 \leq 1,$ &
	$\mbox{Inv-Gamma}_{(0,1)}\left(\frac{1}{2}, \frac{\lambda_1^2}
		{8\sigma^2\lambda_2}\right)$  \\ \\
	\hline\hline \\
	%Kyung \emph{et al.}~(2010) 
	\citet{kyun:10} & 
		$\exp\left\{-\frac{\lambda_2}{2\sigma^2}\beta^T\beta - \frac{\lambda_1}{\sigma}|\beta|_1\right\}$
	& * & * & * \\ \\
	\hline \\
	%Roy and Chakraborty (2017) 
	\citet{roy:17} & $\exp\left\{-\frac{\lambda_2}{2\sigma^2}\beta^T\beta - \frac{\lambda_1}{\sigma}|\beta|_1\right\}$ &
		$\frac{\sigma^2}{\lambda_2}\left(\frac{\lambda_2 \tau_j^2}{1 + \lambda_2 \tau_j^2}\right)$ &
		$\left(1 + \lambda_2 \tau_j^2\right)^{-1/2} e^{-\tau_j^2 \frac{\lambda_1^2}{2}}$, $\tau_j^2 \geq 0$ & $\mbox{UH}\left(1, \frac{1}{2},
			\frac{\lambda_1^2}{2}, \lambda_2\right)$ \\ \\
	\hline\hline \\
\end{tabular}
}
\caption{Prior scaling and data augmentation parameterization in the Bayesian elastic net literature. Double horizontal
	lines differentiate between approaches for scaling the $\ell_1$-norm term in the prior density.
	Entries in the ``SMN variance'' column are the variances of the normal distribution in the scale-mixture-of-normals
	(SMN) representation of the prior; entries in the ``Mixing distribution'' column are the density functions for those random
	variances (with the distributional families named in the final column). ``$\mbox{Gamma}_{(1,\infty)}$'' is a gamma 
	distribution truncated to the interval $(1, \infty)$, ``$\mbox{Inv-Gamma}_{(0,1)}$'' is an inverse gamma distribution truncated 
	to the interval $(0,1)$, and the ``UH'' is distribution is a limit of the compound confluent hypergeometric (CCH) distribution 
	\citep{gord:98}. The SMN representation described in \citet{kyun:10} is inconsistent with the claimed marginal distribution of 
	$\beta$.
	\label{tab:review}}
%\end{sidewaystable}
\end{table}

\subsection{Posterior Computation: Sampling $\beta$}\label{sec:betasamp}
The two main approaches for representing the prior distribution lead to two main approaches for sampling
$\beta$ from its posterior. Under the direct representation of the prior, \citet{hans:11} showed that
the conditional posterior distribution of $\beta$ given $\sigma^2$, $\lambda_1$, and $\lambda_2$ is an 
orthant normal distribution:
\begin{equation}
	\pi(\beta \mid y, \sigma^2, \lambda_1, \lambda_2) = \sum_{z \in \mathcal{Z}} \omega_z
		\mbox{N}^{[z]}\left(\beta \mid \mu_z, \sigma^2 R\right), \label{eq:onpost}
\end{equation}
where the $\omega_z$ are non-negative, orthant-specific weights that sum to one.
The parameters of the underlying normal distributions that generate this posterior have connections
to Bayesian ridge regression \citep{jeff:61, raif:61, hoer:70}. Under both the common and differential 
scalings of the prior, $R = (X^TX + \lambda_2 I_p)^{-1}$ so that $\sigma^2 R$ is the same as the posterior 
covariance matrix for Bayesian ridge regression for a fixed $\lambda_2$.
The orthant-specific location vectors under the commonly-scaled prior are 
$\mu_z = \hat{\beta}_R - \frac{\lambda_1}{2}Rz$, where $\hat{\beta}_R = R X^Ty$ is the 
ridge regression estimate of $\beta$; under the differentially-scaled prior, the location vectors
are $\mu_z = \hat{\beta}_R - \sigma \lambda_1 Rz$ instead. Under both scalings, the orthant-specific weights 
are
$
	\omega_z = \omega^{-1} \frac{\mbox{P}(z, \mu_z, \sigma^2 R)}{\mbox{N}(0 \mid \mu_z, \sigma^2 R)},
$
where $\mbox{P}(z, \mu_z, \sigma^2 R) = \int_{\mathcal{O}_z} \mbox{N}(u \mid \mu_z, \sigma^2 R) du$
and
$
	\omega = \sum_{z \in \mathcal{Z}} \frac{\mbox{P}(z, \mu_z, \sigma^2 R)}{\mbox{N}(0 \mid \mu_z,
	\sigma^2 R)}.
$
A contour plot of the joint posterior density function $\pi_c(\beta \mid y, \sigma^2, \lambda_1, \lambda_2)$
for an example data set when $p=2$ is shown in the left panel of Figure~\ref{fig:postfullcon}. The posterior
is Gaussian within each of the four orthants (quadrants) of $\mathbb{R}^2$, and the density function is 
continuous (bot not differentiable) along the coordinate axes due to the $|\beta|_1$ term in the prior. This
particular example illustrates a situation where the $\lambda_1$ penalty term is large enough that the 
posterior mode lies on one of the coordinate axes ($\beta_2 = 0$).

Sampling directly from (\ref{eq:onpost}) is challenging. An obvious approach is to first sample $z$ from the
discrete distribution over the $2^p$ orthants (each having probability $\omega_z$) and then, conditionally
on the sampled orthant, to sample from a multivariate normal distribution truncated to $\mathcal{O}_z$.
This requires the ability to compute numerically the orthant probabilities $\mbox{P}(z, \mu_z, \sigma^2 R)$ 
and the ability to sample directly from the multivariate truncated normal distribution $\mbox{N}^{[z]}(\beta \mid
\mu_z, \sigma^2 R)$. When $p$ is not too large, the former can sometimes be achieved using, e.g., the 
\texttt{pmvnorm} function in the \texttt{R} package \texttt{mvtnorm} \citep{genz:92, genz:09, genz:24, R},
though if $\lambda_1$ is very large, the probabilities might be too small to compute accurately. The latter,
sampling directly from the truncated multivariate normal distribution, may be difficult even when $p$ is 
small. When $p$ is large, direct sampling from (\ref{eq:onpost}) is not practical.

\citet{hans:11} shows how these issues can be avoided via Gibbs sampling. The full conditional 
posterior distribution for $\beta_j$ is a one-dimensional orthant normal distribution:
\begin{equation}
	\pi(\beta_j \mid y, \beta_{-j}, \sigma^2, \lambda_1, \lambda_2) = 
		(1-\phi_j) \mbox{N}^-(\beta_j \mid \mu_j^-, s^2_j) + 
		\phi_j \mbox{N}^+(\beta_j \mid \mu_j^+, s^2_j).
		\label{eq:bfullcon}
\end{equation}
The scale parameters are $s^2_j = \sigma^2/(x_j^T x_j + \lambda_2)$.
\citet{hans:11} provides an interpretable expression for the location parameter for the positive
component under the commonly-scaled prior (\ref{eq:csprior}):
\begin{equation}
	\mu_j^+ = \hat{\beta}_{R,j} + \left\{ \sum_{i\neq j} \left(\hat{\beta}_{R,i} - \beta_i\right)
		\frac{x_i^T x_j} {x_j^T x_j + \lambda_2}\right\} + \frac{-\lambda_1}{2(x_j^T x_j + \lambda_2)},
		\label{eq:mup}
\end{equation}
where $\hat{R}_{R,i}$ is the $i$th component of the ridge regression estimate $\hat{\beta}_R$ for the 
given value of $\lambda_2$, and $+\lambda_1$ replaces $-\lambda_1$ in the corresponding expression
for $\mu_j^-$. Under the differentially-scaled prior (\ref{eq:altscale}), the trailing term for $\mu_j^{\pm}$ is
$\mp \sigma \lambda_1/(x_j^T x_j + \lambda_2)$ instead.
While (\ref{eq:mup}) has a familiar form---it looks like the usual formula for the 
conditional mean of a Bayesian ridge regression posterior with an additional penalty term
involving $\lambda_1$---a more computationally efficient expression under the commonly-scaled prior 
(\ref{eq:csprior}) is
\[
	\mu_j^+ = \frac{x_j^T y - (x_j^T X_{-j})\beta_{-j} - \lambda_1/2}{x_j^T x_j + \lambda_2},
\]
where $(-\lambda_1/2)$ is replaced with $(+\lambda_1/2)$ in the expression for $\mu_j^-$.
Under the differentially-scaled prior (\ref{eq:altscale}), $\pm \lambda_1/2$ is replaced by 
$\pm \sigma \lambda_1$. The expression is computationally efficient because $X^TX$ and 
$X^Ty$ can be precomputed before the start of the MCMC algorithm. Completing the description of 
the full conditional density, the negative and non-negative components of (\ref{eq:bfullcon}) 
are weighted by
\[
	\phi_j = \left\{ \frac{\Phi(\mu_j^+/s_j)}{\mbox{N}(0\mid \mu_j^+, s^2_j)}\right\} \bigg/ 
		\left\{ 
			\frac{\Phi(\mu_j^+/s_j)}{\mbox{N}(0\mid \mu_j^+, s^2_j)} + 
			\frac{{\Phi(-\mu_j^-/s_j)}}{\mbox{N}(0\mid \mu_j^-, s^2_j)}
		\right\}.
\]

An example full conditional distribution is shown in Figure~\ref{fig:postfullcon}. It is easy to sample
from these distributions. The standard normal cdf, $\Phi(\cdot)$, can be computed to high numerical 
precision (e.g., using the \texttt{pnorm} function in \texttt{R}), and efficient algorithms exist for
sampling from univariate truncated normal distributions 
\citep[e.g., the rejection sampling approach of][]{gewe:91}.

The DA representation of the prior suggests an alternative approach for sampling from the posterior
distribution of $\beta$. \citet{li:10} and \citet{hans:11} describe a two-stage, data augmentation Gibbs
sampler that samples alternately from the conditional posterior of $\beta$ given the latent scale
parameters, $\tau^2$, and then from the conditional posterior of the latent scale
parameters, $\tau^2$, given $\beta$. Under both scalings of the priors, 
the full vector of regression coefficients,
$\beta$, is sampled from the normal distribution $\mbox{N}(\hat{\beta}_{R_*}, \sigma^2 R_*)$,
where $\hat{\beta}_{R_*} = R_* X^T y$. Under the commonly-scaled prior (\ref{eq:csprior}),
$R_*$ has the form $R_c = (X^TX  + \lambda_2 S_{\tau}^{-1})^{-1}$, where 
$S_\tau = \mbox{diag}(1-\tau^2_j)$; under the differentially-scaled prior (\ref{eq:altscale}),
$R_*$ has the form $R_d = (X^TX + D_{\tau}^{-1})^{-1}$, where 
$D_\tau = \mbox{diag}((\tau_j^{-2} + \lambda_2)^{-1})$.
Under the commonly-scaled prior (\ref{eq:csprior}), the latent scale 
parameters are updated by sampling $\zeta_j$ independently from inverse Gaussian distributions with shape 
parameters $\lambda_1^2/(4\lambda_2\sigma^2)$ and means $\lambda_1/(2\lambda_2 |\beta_j|)$,
and then transforming $\tau^2_j = \zeta_j/(1+\zeta_j)$. Under the differentially-scaled prior (\ref{eq:altscale}), 
the latent scale parameters are updated by sampling $\zeta_j$ independently from inverse Gaussian distributions with 
shape parameters $\lambda_1^2$ and means $\sigma \lambda_1/ |\beta_j|$,
and then transforming $\tau^2_j = 1/\zeta_j$.
This sampling scheme has the advantage
that the $\beta_j$ are updated as a block, which may be more effective when the $\beta_j$ are highly
correlated in the posterior. The trade-off for the potential reduction in autocorrelation is the need
to simulate an additional $p$ latent variables.

\begin{figure}[t]
 \begin{center}
 \includegraphics[scale=0.6]{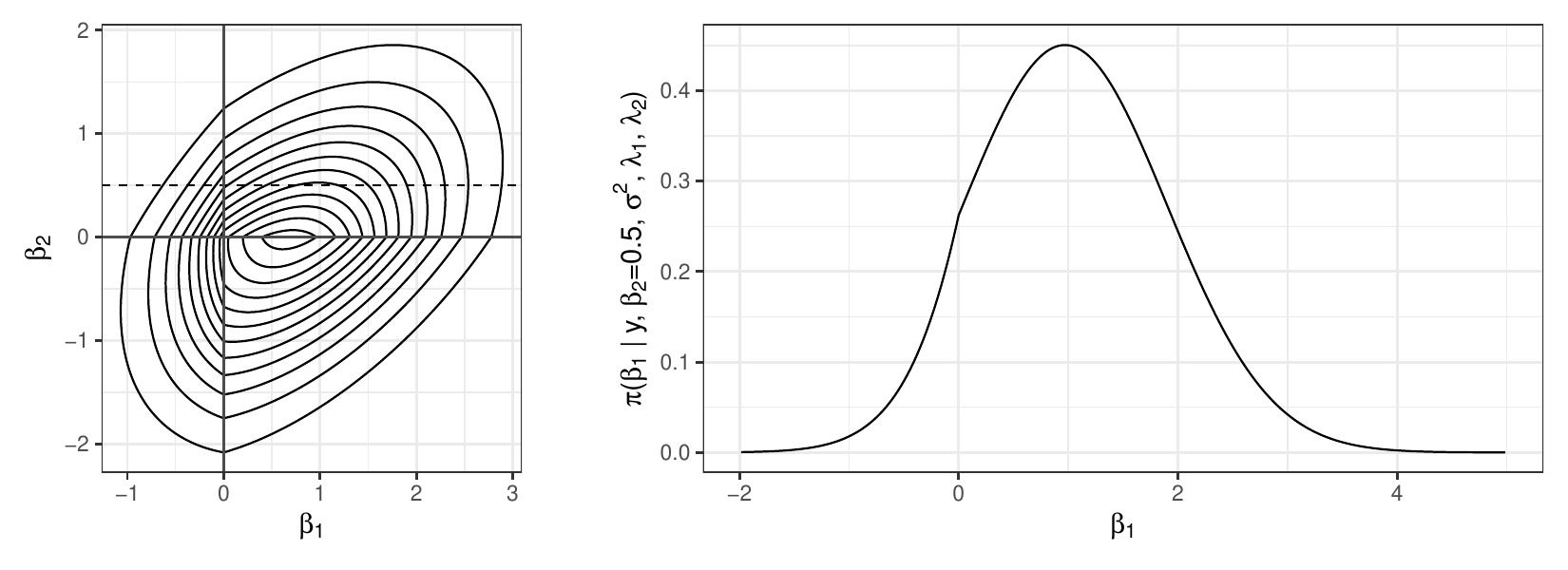}
 \caption{Left: contours of the joint posterior density function 
 	$\pi_c(\beta_1, \beta_2 \mid y, \sigma^2, \lambda_1, \lambda_2)$ for an example data set with
	$p=2$. Right: corresponding full conditional density function $\pi_c(\beta_1 \mid y, \beta_2 = 0.5,
	\sigma^2, \lambda_1, \lambda_2)$ when $\beta_2 = 0.5$ (the dashed line in the plot on the left). 
 	        \label{fig:postfullcon}}
 \end{center}
\end{figure}

\subsection{Posterior Computation: Sampling $\sigma^2$}\label{sec:sig2samp}
A common choice of prior for $\sigma^2$ in Bayesian linear regression is an inverse gamma
distribution, $\mbox{Inv-Gamma}(\nu_a/2,\nu_b/2)$, or its improper limit with $\nu_a = \nu_b = 0$ so that
$\pi(\sigma^2) \propto 1/\sigma^2$.
\citet{kyun:10} \citet{li:10}, \citet{hans:11}, \citet{roy:17}, and \citet{wang:23} all used a version of this
prior in their treatments of the Bayesian elastic net. Under this prior, the form of the full conditional 
distribution for $\sigma^2$ depends on the form and representation for the prior on $\beta$. 
The full conditional has the simplest form under the differentially-scaled prior on $\beta$. As shown in 
\citet{roy:17}, the corresponding full conditional distribution for $\sigma^2$ under the DA representation of 
the differentially-scaled prior on $\beta$ is
\[
	\sigma^2 \mid y, \beta, \tau^2, \lambda_1, \lambda_2 \sim 
		\mbox{Inv-Gamma}\left(\frac{n+p-1+\nu_a}{2}, 
		\frac{\nu_b + (y - X\beta)^T(y - X\beta) + \sum_{j=1}^p \beta_j^2 (\tau_j^{-2} + \lambda_2)}{2}\right),
\]
an inverse gamma distribution that does not depend on $\lambda_1$. 

The full conditional for $\sigma^2$ under the direct representation of the differentially-scaled prior 
has not yet been studied in the literature, but it can be shown that
\[
	\frac{1}{\sigma^2} \mid y, \beta, \tau^2, \lambda_1, \lambda_2 \sim 
		\mbox{MHN}\left(\frac{n+p-1+\nu_a}{2}, 
		\frac{(y - X\beta)^T(y - X\beta) + \lambda_2 \beta^T\beta + \nu_b}{2},
		\lambda_1 |\beta|_1 \right),
\]
a modified half-normal distribution (see Section~\ref{sec:rejsampdeets} for details). Algorithms
for sampling efficiently from this distribution exist \citep[see, e.g.,][]{sun:23}.

The full conditional distribution for $\sigma^2$ under the commonly-scaled prior on $\beta$
is more complex. \citet{hans:11} worked with this full conditional under the direct representation
of the prior, which has density function
\begin{eqnarray*}
	\pi_c(\sigma^2 \mid y, \beta, \lambda_1, \lambda_2) &\propto& 
		\Phi\left(-\frac{\lambda_1}{2\sigma\sqrt{\lambda_2}}\right)^{-p}
		(\sigma^2)^{-(n+p-1+\nu_a)/2 - 1} \times \\
	&& 	\exp\left\{
		-\frac{1}{\sigma^2}\left(
			(y - X\beta)^T(y - X\beta) + \lambda_2 \beta^T\beta + \lambda_1 |\beta|_1
			+ p\lambda_1^2/(4\lambda_2) + \nu_b\right)/2 
			\right\} .
\end{eqnarray*}
This is a non-standard distribution that involves an analytically intractable integral expression,
$\Phi(-\lambda_1/(2\sigma\sqrt{\lambda_2}))$. Noting that as long as $\lambda_1/(2\sigma
\sqrt{\lambda_2})$ is not too large, $\Phi(-\lambda_1/(2\sigma\sqrt{\lambda_2}))$ (or its logarithm) 
can be evaluated numerically to relatively high precision, \citet{hans:11} used a ``Metropolis-within-Gibbs'' 
step to update $\sigma^2$ on its log scale by sampling a proposal, $\log \sigma^{2*}$, from a normal
distribution centered at the current value, $\log \sigma^2$, with a pre-specified innovation variance.
The probability of accepting the proposed value (as opposed to staying at the current value) was
calculated using the ratio of the full conditional for $\log \sigma^2$ evaluated at the proposed and
current values. This algorithm tends to work well in practice, though poor choice of the innovation
variance can result in slow convergence and mixing.

\citet{li:10} and \citet{hans:11} both work with the full conditional for $\sigma^2$ under the DA
representation of the commonly-scaled prior for $\beta$. The full conditional has density function
\begin{eqnarray}
	\pi_c(\sigma^2 \mid y, \beta, \tau^2, \lambda_1, \lambda_2) &\propto& 
		\Phi\left(-\frac{\lambda_1}{2\sigma\sqrt{\lambda_2}}\right)^{-p}
		(\sigma^2)^{-(n+2p-1+\nu_a)/2 - 1} \times \nonumber \\
	&&\exp\left\{
		-\frac{1}{2\sigma^2}\left(\nu_b +
			(y - X\beta)^T(y - X\beta) + 
			\lambda_2\beta^T S_{\tau}^{-1}\beta + 
			\frac{\lambda_1^2}{4\lambda_2}\sum_{j=1}^p \tau_j^{-2}
			\right)
			\right\}. \label{eq:dacssig2}
\end{eqnarray}
\citet{hans:11} used a random-walk Metropolis algorithm for updating $\log \sigma^2$.
\citet{li:10} reexpressed this density in terms of the upper incomplete gamma function
\citep{dido:86}, $\Gamma_U(\alpha,x) = \int_x^\infty t^{\alpha-1}e^{-t}dt$, through the 
equivalence $2\sqrt{\pi}\Phi\left(\frac{-\lambda_1}{2\sigma\sqrt{\lambda_2}}\right) =
\Gamma_U\left( \frac{1}{2}, \frac{1}{2}
\left(\frac{\lambda_1}{2\sigma\sqrt{\lambda_2}} \right)^2\right)$, and proposed
obtaining exact samples from this full conditional distribution via rejecting sampling
using an inverse gamma proposal distribution. Unfortunately, the derivation of the
acceptance probability for the algorithm contains an error and the resulting samples
do not come from the desired target distribution. We identify the problem in detail
in Appendix~\ref{app:sig2rejsamp}.

\subsection{Posterior Computation: Sampling $\lambda_1$ and $\lambda_2$}\label{sec:othersamp}

Posterior sampling of $\lambda_1$ and $\lambda_2$ is non-trivial under both scalings of the prior
on $\beta$ whether or not data augmentation is used. The term 
$\Phi(-\lambda_1/(2\sigma \sqrt{\lambda_2}))^{-p}$ (under prior (\ref{eq:cscalebone})) or
$\Phi(-\lambda_1/\sqrt{\lambda_2})^{-p}$ (under prior (\ref{eq:dscalebone})) appears in the posterior,
making direct sampling of these parameters difficult. \citet{li:10} and \citet{roy:17} eschew full
Bayesian inference and instead propose methods for selecting values for these hyperparameters.
\citet{hans:11} uses separate random-walk Metropolis updates for $\log \lambda_1$ and
$\log \lambda_2$ with their full conditionals as the target densities. This approach requires 
specification of a step-size parameter for the normal proposal, which can be
difficult to select and tune. Motivated by the desire to avoid numerical computation of the
standard normal cdf, $\Phi(\cdot)$, \citet{wang:23} devise a clever exchange algorithm 
\citep{murr:06} that introduces $p$ additional latent variables in such a way as to remove
the term involving $\Phi(\cdot)$ from the joint posterior of the augmented parameter space.
Despite avoiding computation of $\Phi(\cdot)$, the algorithm still requires one parameter
to be updated via the Metropolis--Hastings algorithm using a random-walk proposal,
necessitating the selection of a step-size parameter.

All of the correctly-specified approaches reviewed above for full Bayesian inference for the 
Bayesian elastic net require the specification of at least one step-size parameter in
a Metropolis--Hastings step in a Gibbs sampler. This can be challenging in practice,
as appropriate scales for the step sizes are not always obvious before running the
MCMC algorithm. If poor step sizes are chosen, the resulting Markov chains will
mix slowly. Practitioners who make use of appropriate post-sampling MCMC diagnostics might 
notice this, adjust the step size, and rerun the chain (perhaps iterating this procedure several 
times); practitioners who simply use the original MCMC output as is will produce low-quality 
summaries of the posterior.

We introduce in Section~\ref{sec:rejsamp} a new approach to posterior sampling for full 
Bayesian inference for the Bayesian elastic net that avoids these issues entirely. We use
a simple transformation of the parameter space to (i) reduce the number of 
parameters whose full conditional densities have a term involving $\Phi(\cdot)$ and (ii) produce
log-concave full conditional distributions that can be easily sampled via a highly-efficient
rejection sampling algorithm using automatically-tuned, piece-wise exponential proposal 
distributions. Importantly, the approach requires no tuning on the part of the analyst.

\section{Efficient Rejection Sampling for Full Bayesian Inference}\label{sec:rejsamp}
Full Bayesian inference for Bayesian elastic net regression proceeds by assigning prior
distributions to $\sigma^2$, $\lambda_1$, and $\lambda_2$, and making inferences
based on the joint posterior $\pi(\beta, \sigma^2, \lambda_1, \lambda_2 \mid y)$ and
its margins. \citet{li:10}, \citet{hans:11}, and \citet{roy:17} all assign to $\sigma^2$ an 
inverse gamma prior, $\sigma^2 \sim \mbox{Inv-Gamma}(\nu_a/2, \nu_b/2)$, or the improper
prior with $\pi(\sigma^2) \propto \sigma^{-2}$. Under the commonly-scaled prior for 
$\beta$, \citet{hans:11} considered the
hyperprior distributions $\lambda_1 \sim \mbox{Gamma}(L, \nu_1/2)$ and
$\lambda_2 \sim \mbox{Gamma}(R, \nu_2/2)$, where the gamma distributions are
parameterized to have mean $2L/\nu_1$ and $2R/\nu_2$. When the prior is
parameterized as in (\ref{eq:csprior}), \cite{zou:05} noted that 
$\lambda = \lambda_1+ \lambda_2$ represents the total penalization and 
$\alpha = \lambda_2/(\lambda_1+\lambda_2)$ represents the proportion of the total 
penalization that is attributable to the $\ell_2$-norm component. As noted in \citet{hans:11}, 
when the gamma priors on $\lambda_1$ and $\lambda_2$ are independent and 
$\nu_1 = \nu_2 = \nu$, the induced priors on the transformed parameters are 
$\lambda \sim \mbox{Gamma}(R+L, \nu/2)$ and $\alpha \sim \mbox{Beta}(R, L)$,
with $\lambda$ and $\alpha$ independent. Considering the hyperpriors from these
two perspectives gives the user a range of interpretations to consider when
specifying $L$, $R$, $\nu_1$, and $\nu_2$. We use these prior distributions
for $\sigma^2$, $\lambda_1$, and $\lambda_2$ unless noted otherwise.

\subsection{Rejection sampling for a class of distributions}\label{sec:rejsampdeets}
Our first approach to posterior sampling, described in Section~\ref{sec:trans}, requires
the ability to generate random variates from density functions of the form
\begin{eqnarray}
	f(x) \propto \Phi(-x)^{-q} x^{a-1}e^{-bx^2 - cx - d/x}, \; x > 0,
		\label{eq:f(x)}
\end{eqnarray}
where $q \in \{0, 1, \ldots \}$. The function $f$ is a density function under various
conditions on $q$, $a$, $b$, $c$, and $d$. Three conditions are of special
interest to us.

First, the conditions $\{q = 0, a \in \mathbb{R}, b = 0, c>0, d > 0\}$ correspond
to the family of generalized inverse Gaussian (GIG) distributions \citep{barn:77, barn:78}.
\citet{devr:14} introduced a rejection sampling algorithm for sampling GIG random variates 
$X$ based on log-concavity of the density of $\log X$. 
Second, the conditions $\{q = 0, a > 0, b > 0, c  \in \mathbb{R}, d = 0\}$
correspond to the family of modified half normal (MHN) distributions.
\citet{sun:20} and \citet{sun:23} introduce efficient rejection sampling algorithms for obtaining samples
from this class of distributions.

We can use the rejection sampling algorithms of \citet{devr:14} and \citet{sun:23} to
obtain samples from the GIG and MHN distributions, respectively. When $a \geq 1$,
the GIG and MHN distributions both have log concave densities, in which case we could
also use rejection sampling techniques that exploit log concavity. We describe such an 
approach here that can be used to obtain samples from (\ref{eq:f(x)}) when $q=0$,
$a \geq 1$, and $b$, $c$, and $d$ are such that $f(x)$ is integrable and log concave.
The approach is strongly connected
to the work of \citet{devr:84, devr:86}, \citet{gilk:92a}, and \citet{gilk:92b} in the
sense that a piece-wise exponential hull is used to bound the target density,
with proposals drawn from the corresponding piece-wise exponential
distribution. The approach is not ``adaptive'' in the sense that the piece-wise
exponential hull is not refined if a proposal is rejected, but it is ``adapted''
to the target density because information about the log density's mode and
curvature at the mode are used to construct the proposal density. The approach
was used by \citet{hans:09} to sample $\sigma^2$ for Bayesian lasso regression.

To construct the piece-wise exponential hull, we make use of the unique mode,
$x_*$, of the distribution. For the GIG distribution with $a \geq 1$, the unique mode 
occurs at  $(a - 1 + \sqrt{(a-1)^2 + 4cd})/(ac)$; for the MHN distribution with $a\geq 1$, 
the unique mode occurs at $(-c + \sqrt{c^2 + 8b(a-1)})/(4b)$.
The piece-wise exponential hull is created by 
first placing a knot point at the mode, $x_*$. Several knot points are then placed at 
appropriate distances above and below the mode. To determine where to place these
additional knot points, we make use of the curvature of $\log f$ at is mode, 
$f^{\prime \prime}(x_*) = -(a-1)/x_* - 2b - 2d/x_*^3$, by noting that if $f$ was the
density for a Gaussian distribution,
$s_{x_*} = |f^{\prime\prime}(x_*)|^{-1/2}$ would be the standard deviation of the
distribution and would provide a scale to inform us about where to place the knots. 
Using this second-order Taylor polynomial approximation to $\log f(x)$, we place one knot
at $x_* + s_{x_*}/2$, and then $K$ additional knots at 
$x_* + k s_{x_*}$, $k = 1, \ldots, K$. Below the mode, we place knots at 
$x_* - s_{x_*}/2$ and $x_* - k s_{x_*}$, $k = 1, \ldots, K$. Any negative knots are then
removed from the set; if no knots remain below the mode, a single knot is then
placed at $x_*/2$.
Lines tangent to $\log f$ at the knot points are then used to construct a piece-wise
linear upper hull for $\log f$, with change points occurring at the intersections
of the tangent lines. The piece-wise linear upper hull is exponentiated to obtain
a piece-wise exponential hull for $f$, which can then be rescaled and used as a
proposal distribution for rejection sampling. The key elements of the approach
are depicted graphically in Figure~\ref{fig:rejsamp}. Practical experience suggests that a 
small number of knot points ($K=2$ or $K=3$) results in high acceptance rates with 
low computational overhead.

\begin{figure}[t]
 \begin{center}
 \includegraphics[scale=0.55]{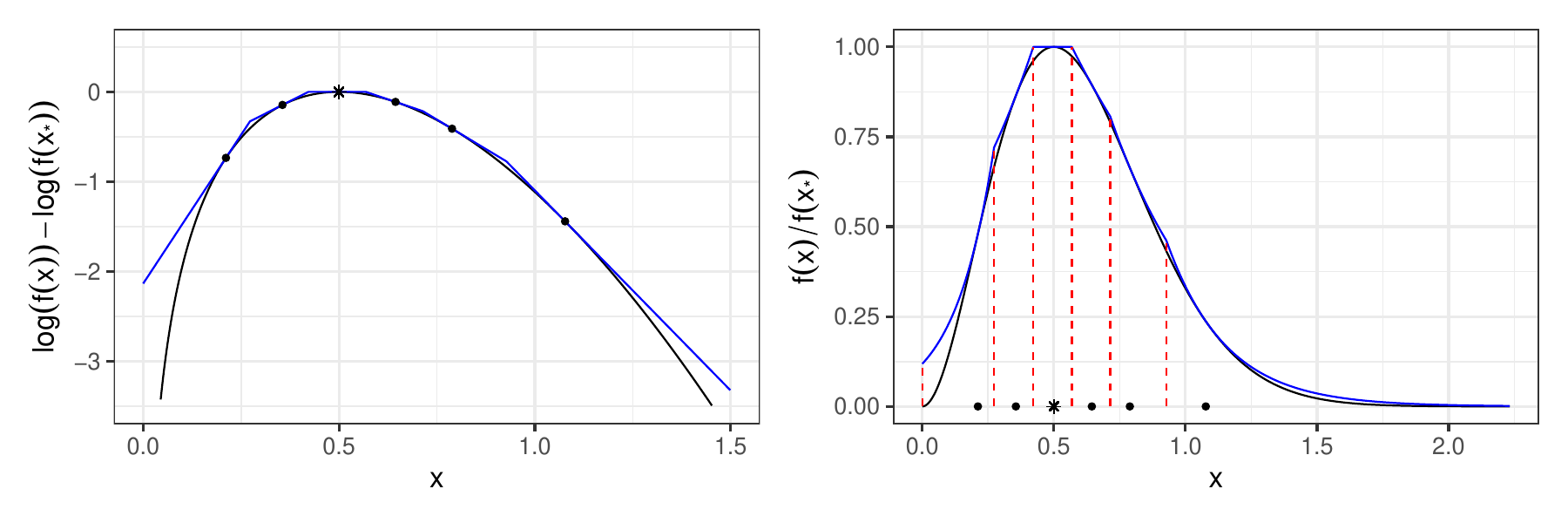}
 \caption{Illustration of the rejection sampling appraoch described in
 		Section~\ref{sec:rejsampdeets} for a target density $f(x)$ with $q=0$,
		$a=3$, $b=2$, $c=2$, and $d=0$ (a modified half normal distribution). 
		The left panel displays the 
		piece-wise linear hull (blue line) for $\log f(x)$ (black line), 
		where both functions are shifted to take the value 0 at the mode,
		$x_* = 0.5$. The right panel displays the piece-wise exponential
		hull (blue line) for $f(x)$ (black line), where both functions are scaled to take
		the value 1 at the mode. The black points indicate the locations of the 
		``knot points'', with  ``*'' corresponding to the mode, $x_*$.
		The vertical red dashed lines partition the support of $f$ according to
		the intersections of the lines tangent to $f$ at the knot points and define
		the change points for the piece-wise exponential proposal distribution.
		The rejection sampler for this example has an acceptance probability of 
		approximately $0.954$.
 	        \label{fig:rejsamp}}
 \end{center}
\end{figure}

The third set of conditions under 
which we will need to sample from $f(x)$ is $\{q \in \{1,2, \ldots\}, a > 0, b \geq q/2, c > 0, d = 0\}$.
This density function is more challenging because it contains the nontrivial term
$\Phi(-x)^{-q}$. Under the more strict condition that $a \geq 1$, we can show
that $f(x)$ is log concave, and a modified version of our rejection sampling
algorithm can be used to obtain samples from $f(x)$.

\begin{proposition}\label{prop:logconcave}
The function $f(x) \propto \Phi(-x)^{-q} x^{a-1}
e^{-bx^2 - c x}$, $x > 0$, is integrable and log concave
when $q \in \{1, 2, \ldots\}$, $ a \geq 1$, $b \geq q/2$, and $c > 0$.
\end{proposition}

\begin{proof}
It is clear that $\int_0^t f(x) dx < \infty$ for
all $t < \infty$ because $2^q \leq \Phi(-x)^{-q} < \infty$
for all $0 \leq x \leq t < \infty$. By \citet{fell:68},
$\Phi(-x) \geq (2\pi)^{-1/2}(x^{-1} - x^{-3}) e^{-x^2/2}$,
and so for large enough $x$, $\Phi(-x)^{-q} \leq (2\pi)^{q/2}
x^{3q} e^{qx^2/2}$. For large enough $t$,
\begin{eqnarray*}
	\int_t^\infty \Phi(-x)^{-q} x^{a-1} e^{-bx^2 - c x}
		dx &=& 
		\int_t^\infty \Phi(-x)^{-q} e^{-qx^2/2} x^{a-1} e^{-(b-q/2)x^2 - c x}
		dx \\
	&\leq&  (2\pi)^{q/2} \int_t^\infty  x^{3q+a-1} 
		e^{- (b - q/2)x^2 -c x} dx \\
	&<& \infty,
\end{eqnarray*}
and so $f(x)$ is integrable.

Now let $Z \sim \mbox{N}(0, I_q)$ be a $q$-vector of independent, 
standard normal random variables, and let $|Z|_1 = \sum_{j=1}^q
|Z_j|$. For $x > 0$, it can be shown
that $\Ex{e^{-x |Z|_1}} = 2^q e^{qx^2/2} \Phi(-x)^q$.
We can then write
$
	f(x) \propto \Ex{e^{-x |Z|_1}}^{-1}
		x^{a-1}e^{-(b-q/2)x^2-cx}.
$
The term $x^{a-1}e^{-(b-q/2)x^2-cx}$ is log concave when $a \geq 1$,
$b \geq q/2$, and $c > 0$,
and so we need only show log convexity of $h(x) \equiv \Ex{e^{-x |Z|_1}}$.
For any $\alpha \in [0, 1]$, $x_1 >0$, and $x_2 > 0$,
\begin{eqnarray*}
	h((1-\alpha) x_1 + \alpha x_2) &=& \int_{\mathbb{R}^q}
		e^{-((1-\alpha)x_1 + \alpha x_2)|z|_1} \mbox{N}(z \mid
		0, I_q) dz \\
	&=& \int_{\mathbb{R}^q} \left( 
		e^{-x_1 |z|_1} \mbox{N}(z \mid 0, I_q)
		\right)^{1-\alpha}
		\left(
		e^{-x_2|z|_1} \mbox{N}(z \mid
		0, I_q)\right)^{\alpha} dz \\
	&\leq& \left\{\int_{\mathbb{R}^q} \left( 
		e^{-x_1 |z|_1} \mbox{N}(z \mid 0, I_q)
		\right)dz \right\}^{1-\alpha}
		\left\{\int_{\mathbb{R}^q} \left( 
		e^{-x_2 |z|_1} \mbox{N}(z \mid 0, I_q)
		\right)dz \right\}^{\alpha} \\
	&=& h(x_1)^{1-\alpha}h(x_2)^\alpha,
\end{eqnarray*}
where the inequality is due to H\"older's inequality.
The function $h(x)$ is therefore log convex and hence $f(x)$
is log concave.
\end{proof}

To implement the same rejection sampling algorithm as above, we need to be able to find the
mode of $\log f(x)$, which has derivative
\[
	\frac{d}{dx}\log f(x) = q \frac{\phi(-x)}{\Phi(-x)} + \frac{a-1}{x} - 2bx - c.
\]
The mode is no longer available as the unique, positive root of a polynomial.
A mode-finding algorithm could be used to approximate $x_*$ at some additional computational
cost. Instead, we note that in order to construct an integrable, piece-wise exponential
hull for $f(x)$, we need only identify at least one knot point above the mode because the support
of $x$ is bounded below by zero; additionally identifying at least one point below the mode 
will help improve the quality of the piece-wise exponential approximation.
We can facilitate the choice of such knot points using the following result.

\begin{proposition}\label{prop:modebound}
When $q \in \{1, 2, \ldots \}$, $ a \geq 1$, $b \geq q/2$, and $c > 0$,
the function $f(x) \propto \Phi(-x)^{-q} x^{a-1}
e^{-bx^2 - c x}$, $x > 0$, has a unique mode, $x_*$,
satisfying 
\[
	%\left\{
	\begin{array}{cl}
		\frac{a-1}{c} < x_* < \frac{a-1+q}{c}, & \mbox{if } a \geq 1 , q = 2b, \\	
%		0 < x_* < \frac{q}{c}, & 
%			\mbox{if } a = 1, q = 2b, \\ 
		0 < x_* < \frac{\sqrt{c^2 + 4q(2b-q)} - c}{2(2b-q)}, & 
			\mbox{if } a = 1, q > 2b, \mbox{ and}\\
		\frac{\sqrt{c^2 + 4(a-1)(2b-q)} - c}{2(2b-q)} < x_* <
		\frac{\sqrt{c^2 + 4(a-1+q)(2b-q)} - c}{2(2b-q)}, & \mbox{if } a > 1, q > 2b.
	\end{array} %\right.
\]
\end{proposition}

\begin{proof}
The function $f(x)$ has a unique mode because it is log concave, and
the derivative
\[
	\frac{d}{dx}\log f(x) = q \frac{\phi(-x)}{\Phi(-x)} + \frac{a-1}{x} - 2bx - c
\]
is decreasing in $x$ for the same reason. By \citet{gord:41}, we have
\[
	x < \frac{\phi(-x)}{\Phi(-x)} < x + \frac{1}{x},
\]
for $ x > 0$, and so we can bound the derivative of $\log f(x)$ as
\[
	g_L(x) \equiv \frac{a-1}{x} - (2b-q)x - c < \frac{d}{dx} \log f(x) < \frac{a-1+q}{x} - (2b-q)x - c
	\equiv g_U(x).
\]

When $b = q/2$ and $a>1$, $g_L(x)$ and $g_U(x)$ are both
are decreasing functions. The sign of $\frac{d}{dx} \log f(x)$ is therefore positive when
$x < (a-1)/c$ due to $g_L(x)$ and negative when $x > (a-1+q)/c$ due to $g_U(x)$ and so, by the 
intermediate value theorem, $(a-1)/c < x_* < (a-1+q)/c$.

Next, when $b=q/2$ and $a=1$, $g_L(x) = -c < 0$ for all $x > 0$ and so the lower bound is not useful. The upper bound, 
$g_U(x)$, is decreasing, and so $\frac{d}{dx} \log f(x)$ is negative when $x > q/c$. By the intermediate value 
theorem, we must have $0 < x_* < q/c$ because $\lim_{x\rightarrow 0^+} \frac{d}{dx} \log f(x) = \infty$.

Now focusing on the case $b > q/2$, when $a=1$ we have $g_L(x) = -(2b-q)x - c < 0$ for all $x > 0$, 
and so the lower bound is not useful. The upper bound is $g_U(x) = q/x - (2b-q)x - c$, a decreasing function. The sign of
$\frac{d}{dx} \log f(x)$ is therefore negative when $x > \frac{\sqrt{c^2 + 4q(2b-q)} - c}{2(2b-q)}$,
and so by the same arguments as above we must have $0 < x_* < \frac{\sqrt{c^2 + 4q(2b-q)} - c}{2(2b-q)}$.

Finally, when $b>q/2$ and $a>1$, $g_L(x) = (a-1)/x - (2b-q)x - c$ and $g_U(x) = (a-1+q)/x - (2b-q)x - c$, 
and both are decreasing functions. The same arguments as above yields the resulting bounds on $x_*$.

\end{proof}

We can therefore find one point above the mode and, when $a > 1$, one point below the mode that
can be used to construct a piecewise linear upper hull for $\log f(x)$. With only one or two knot
points which might be ill-positioned depending on the quality of the bound(s),
the resulting piece-wise exponential proposal distribution may result in 
a high rejection rate, requiring either additional well-placed knot points to start or a strategy for
adapting the hull as proposals are rejected. Given the difficulty of finding a suitable set of knot
points, we instead use the traditional adaptive rejection sampling algorithm \citet{gilk:92b}
to sample from $f(x)$ when $q > 0$. In practice, we use the \texttt{ars} package \citep{paul:24}
in \texttt{R} \citep{R}, supplying the \texttt{ars} function with inputs $\log f(x)$, $\frac{d}{dx} \log f(x)$,
a lower bound of $x=0$, and one (or two) initial knot points above (and below) the mode.

\subsection{Rejection sampling for full Bayesian inference}\label{sec:trans}
As described in Section~\ref{sec:othersamp}, the full conditional density functions
for $\lambda_1$ and $\lambda_2$ under the differentially-scaled prior are
not available in closed form due to the $\Phi(-\lambda_1/\sqrt{\lambda_2})$
term in (\ref{eq:altdirect}) under the direct representation or in
(\ref{eq:altda}) under the data augmentation representation.
The same is true for $\sigma^2$ under the commonly-scaled prior due to the
$\Phi(-\lambda_1/(2\sigma\sqrt{\lambda_2}))$ term in (\ref{eq:onform})
and (\ref{eq:directda}).  As an improvement to existing MCMC methods for full Bayesian inference,
we consider a transformation of the parameter space that (i) confines the
awkward $\Phi(\cdot)$ term to a single full conditional distribution and (ii) results in
log-concave full conditional density functions for all parameters that do not have
``standard'' (easy to sample from) full conditional distributions. We then exploit 
log-concavity to construct efficient rejection sampling algorithms for these parameters.
The form of the transformation depends on whether the commonly- or
differentially-scaled prior for $\beta$ is used. We start with the commonly-scaled
prior.

\subsubsection{Sampling under the commonly-scaled prior}\label{sec:sampcom}
Under prior (\ref{eq:cscalebone}), define the transformation $(\sigma^2, \lambda_1, 
\lambda_2) \rightarrow (u_1 = \sigma^2, u_2 = \sqrt{\lambda_2}/\sigma,
\theta = \lambda_1/(2\sigma \sqrt{\lambda_2}))$. The reparameterized prior
on the regression coefficients is then
\[
		\pi_c(\beta \mid u_1, u_2, \theta) = 
		2^{-p} (2\pi)^{-p/2}  u_2^p
		e^{-\frac{p\theta^2}{2}}
		 \Phi\left(-\theta \right)^{-p}
		\exp\left\{-u_2^2 \beta^T \beta/2 - u_2 \theta |\beta|_1 \right\}.
\]
The awkward term involving $\Phi(\cdot)$ is now a function of only
$\theta$. Transforming the prior on $\sigma^2$, $\lambda_1$,
and $\lambda_2$ yields
\begin{equation}
	\pi(u_1, u_2, \theta) \propto 
		u_1^{R  + L - \nu_a/2 - 1} u_2^{2R + L - 1} \theta^{L-1}
		\exp\left\{ - u_1 u_2^2 \nu_2/2 - u_1 u_2 \theta \nu_1 - u_1^{-1} \nu_b/2
		\right\}. \label{eq:transprior}
\end{equation}
Combining these priors with the likelihood function yields the following
full conditional posterior distributions for $u_1$, $u_2$,
and $\theta$:
\begin{eqnarray*}
	u_1 \mid y, \beta, u_2, \theta &\sim& \mbox{GIG}\left( R + L - (\nu_a + n - 1)/2,
		u_2^2 \nu_2 + 2u_2 \theta \nu_1, 
		(y - X\beta)^T(y - X\beta) + \nu_b\right), \\
	u_2 \mid y, \beta, u_1, \theta &\sim& \mbox{MHN}\left( 2R + L + p,
		\frac{u_1 \nu_2 + \beta^T\beta}{2}, \theta (u_1 \nu_1 + |\beta|_1)\right), \\
	\pi_c(\theta \mid y, \beta, u_1, u_2) &\propto& 
		\Phi\left(-\theta\right)^{-p}
		\theta^{L-1}
		\exp\left\{-p\theta^2/2 - \theta u_2(u_1 \nu_1 + |\beta|_1 )
		\right\}.
\end{eqnarray*}
We obtain samples from these distributions as follows.

The full conditional posterior distribution for $u_1$ is a generalized inverse
Gaussian (GIG) distribution \citep{barn:77, barn:78}. As discussed in 
Section~\ref{sec:rejsampdeets}, we can use the rejection sampling method
of \citet{devr:14} to sample $u_1$ from its full conditional.
The full conditional for $u_1$ will be log concave when $R + L - (\nu_a + n - 1)/2 \geq 1$,
in which case we could also use the method described in Section~\ref{sec:rejsampdeets}
that uses information about the full conditional at its mode to sample from this
distribution.

The full conditional posterior distribution for $u_2$ is a modified half normal
(MHN) distribution and will always be log concave as we assume $R > 0$ and $L > 0$ 
in the prior. We can therefore use either the rejection sampling method introduced
by \citet{sun:23} or the rejection sampling method described in Section~\ref{sec:rejsampdeets}
that uses information about the full  conditional at its mode to sample from the full
conditional. 

The full conditional for $\theta$ has the form of (\ref{eq:f(x)}) with $q=p \in \{1, 2, \ldots,\}$,
$b = p/2$, and $c > 0$. By Proposition~\ref{prop:logconcave}, the full conditional
will be log concave as long as $L \geq 1$ in the prior on $\lambda_1$, and we can 
use adaptive rejection sampling by identifying at least one knot point to the right of the
distribution's mode via Proposition~\ref{prop:modebound}. 
The full conditional is not log concave when $0 < L < 1$, in which case other methods for
sampling $\theta$ would be required. While this might be considered a limitation to our approach 
to sampling, we note that the rejection sampling method described below under the DA 
representation of the prior requires only that $L > 0$, and so we can always simply use the
DA Gibbs sampler when $0 < L < 1$.

Focusing now on the DA representation of the prior, under the commonly-scaled prior
(\ref{eq:csprior}) and the same transformation $(\sigma^2, \lambda_1, 
\lambda_2) \rightarrow (u_1 = \sigma^2, u_2 = \sqrt{\lambda_2}/\sigma,
\theta = \lambda_1/(2\sigma \sqrt{\lambda_2}))$, the reparameterized joint prior
on $\beta$ and $\tau$ is
\begin{eqnarray}
		\pi_c(\beta, \tau \mid u_1, u_2, \theta) &=& \pi_c(\beta \mid \tau, u_1, u_2, \theta)
			\pi_c(\tau \mid u_1, u_2, \theta) \nonumber \\ 
		&\propto& \Phi\left(-\theta\right)^{-p} \theta^p  u_2^p \left[\prod_{j=1}^p \tau_j^{-3/2}(1-\tau_j)^{-1/2} \right]
			\times \nonumber \\
		&&  \exp\left\{-\frac{u_2^2}{2} \beta^T S_\tau^{-1} \beta -\frac{\theta^2}{2}\sum_{j=1}^p \tau_j^{-1} \right\}. 
			\label{eq:transbcdaprior}
\end{eqnarray}
The awkward term involving $\Phi(\cdot)$ is now a function of only $\theta$. Combining the 
transformed prior (\ref{eq:transprior}) with (\ref{eq:transbcdaprior}) and the likelihood yields
the following full conditional distributions:
\begin{eqnarray*}
	u_1 \mid y, \beta, \tau, u_2, \theta &\sim& \mbox{GIG}\left( R + L - (\nu_a + n - 1)/2, 
		u_2^2 \nu_2 + 2 u_2 \theta \nu_1, (y - X\beta)^T(y - X\beta) + \nu_b \right), \\
	u_2 \mid y, \beta, \tau, u_1, \theta &\sim& \mbox{MHN}\left(2R + L + p, \frac{u_1\nu_2 + \beta^T S_\tau^{-1}\beta}{2},
		u_1\theta \nu_1 \right), \\
	\pi_c(\theta \mid y, \beta, \tau, u_1, u_2) &\propto&
		\Phi\left(-\theta\right)^{-p} \theta^{p+L-1} \exp\left\{-\frac{\theta^2}{2} \left(\sum_{j=1}^p \tau_j^{-1}\right) -
			\theta u_1u_2\nu_1\right\}.
\end{eqnarray*}
As above, we use the method of \citet{devr:14} to sample from the inverse Gaussian full conditional for
$u_1$. The full conditional for $u_2$ will always be log concave and so we use either the method of
\citet{sun:23} or the rejection sampling method described in Section~\ref{sec:rejsampdeets} to update $u_2$.
The full conditional for $\theta$ has the form of (\ref{eq:f(x)}) with $q=p \in \{1, 2, \ldots,\}$,
$a = p + L > 1$, $b = \sum_{j=1}^p \tau_j^{-1}/2$, and $c > 0$. Because $0 < \tau_j < 1$, we have
$b = \sum_{j=1}^p \tau_j^{-1}/2 > p/2$, and by Proposition~\ref{prop:logconcave} the full 
conditional is log concave and we can sample from this distribution using adaptive rejection sampling by identifying 
at least one knot point to the right of the distribution's mode via Proposition~\ref{prop:modebound}.

\subsubsection{Sampling under the differentially-scaled prior}\label{sec:sampdiff}
Under the differentially-scaled prior (\ref{eq:dscalebone}), consider the transformation $(\lambda_2, \lambda_1) \rightarrow
(u_2 = \sqrt{\lambda_2}, \theta = \lambda_1/\sqrt{\lambda_2})$. The reparameterized prior on the regression coefficients
is
\[
	\pi_d(\beta \mid \sigma^2, u_2, \theta) = 2^p (2\pi)^{-p/2} (\sigma^2)^{-p/2} u_2^p e^{-p\theta^2/2} \Phi(-\theta)^{-p}
		\exp\left\{-\frac{u_2^2}{2\sigma^2}\beta^T\beta - \frac{\theta u_2}{\sigma}|\beta|_1\right\},
\]
the prior on $\sigma^2$ remains an inverse gamma distribution (or its improper limit), and the prior on the transformed
parameters is
\begin{equation}
	\pi(u_2, \theta) \propto u_2^{2R + L - 1} \theta^{L-1} e^{-u_2^2 \nu_2/2 - u_2 \theta \nu_1/2}.
	\label{eq:transdprior}
\end{equation}
Combining these priors with the likelihood function yields the following full conditional posterior distributions:
\begin{eqnarray*}
	\frac{1}{\sigma^2} \mid y, \beta, u_2, \theta &\sim& \mbox{MHN}\left( \frac{\nu_a + p + n - 1}{2},
		\frac{(y-X\beta)^T(y-X\beta) + u_2^2 \beta^T\beta + \nu_b}{2}, \theta u_2 |\beta|_1\right), \\
	u_2 \mid y, \beta, \sigma^2, \theta &\sim& \mbox{MHN}\left( 2R + L + p, \frac{\beta^T\beta/\sigma^2
		 + \nu_2}{2}, \theta\left(|\beta|_1/\sigma + \nu_1/2\right) \right), \\
	\pi_d(\theta \mid y, \beta, \sigma^2, u_2) &\propto& \Phi\left(-\theta\right)^{-p} \theta^{L-1}
		\exp\left\{-\theta^2 p/2 - \theta u_2 \left(|\beta|_1/\sigma + \nu_1/2\right) \right\}.
\end{eqnarray*}
The full conditionals for $\sigma^{-2}$ and $u_2$ are log concave, and so we can either use the rejection
sampling methods described in Section~\ref{sec:rejsampdeets} or the method of \citet{sun:23} to obtain
samples from the full conditionals. The full conditional for $\theta$ will be log concave when $L \geq 1$,
in which case we can use rejection sampling as described in Section~\ref{sec:rejsampdeets}. When
$0 < L < 1$, the full conditional is not log concave and we cannot use this particular method of rejection
sampling. As in Section~\ref{sec:sampcom}, when $0 < L < 1$ we can instead implement a DA Gibbs sampler
as described below.

Under the DA representation of the differentially-scaled prior
and the same transformation $(\lambda_2, \lambda_1) \rightarrow (u_2 = \sqrt{\lambda_2},
\theta = \lambda_1/\sqrt{\lambda_2})$, the reparameterized joint prior on $\beta$ and $\tau^2$ is
\begin{eqnarray}
		\pi_d(\beta, \tau^2 \mid \sigma^2, u_2, \theta) &=& \pi_d(\beta \mid \tau^2, \sigma^2, u_2, \theta)
			\pi_d(\tau^2 \mid \sigma^2, u_2, \theta) \nonumber \\ 
		&\propto& (\sigma^2)^{-p/2} \Phi\left(-\theta\right)^{-p} \theta^p u_2^{2p}
			\left[\prod_{j=1}^p (\tau_j^2)^{-1/2}\right] \times \nonumber \\
		&& \exp\left\{ -\frac{\theta^2}{2} \left( p + u_2^2\sum_{j=1}^p \tau_j^2\right) -
			\frac{1}{2\sigma^2} \sum_{j=1}^p \beta_j^2\left(\tau_j^{-2} + u_2^2\right) \right\}.
			\label{eq:transbddaprior}
\end{eqnarray}
Combining the transformed prior (\ref{eq:transdprior}) with (\ref{eq:transbddaprior}) and the likelihood yields
the following full conditional posterior distributions, which can be sampled as described above:
\begin{eqnarray*}
	\sigma^2 \mid y, \beta, \tau^2, u_2, \theta &\sim& \mbox{Inv-Gamma}\left(\frac{p+\nu_a+n-1}{2},
		\frac{\nu_b + (y - X\beta)^T(y - X\beta) + \sum_{j=1}^p \beta_j^2 (\tau_j^{-2} + u_2^2)}{2} \right), \\
	u_2 \mid y, \beta, \tau^2, \sigma^2, \theta &\sim& \mbox{MHN}\left(2p + 2R + L, 
		\frac{\beta^T\beta/\sigma^2 + \nu_2 + \theta^2 \sum_{j=1}^p \tau_j^2}{2}, \frac{\theta\nu_1}{2}\right), \\
	\pi_d(\theta \mid y, \beta, \tau^2, \sigma^2, u_2) &\propto&
		\Phi\left(-\theta\right)^{-p} \theta^{p+L-1} \exp\left\{-\frac{\theta^2}{2}\left(p + u_2^2\sum_{j=1}^p \tau_j^2\right)
			- \theta u_2\nu_1/2\right\}.
\end{eqnarray*}

\section{Simulations}\label{sec:sim}
We conduct a simulation study to document the relative performance of the various approaches to MCMC 
for full Bayesian elastic net inference under a several scenarios for data generation. The existing, 
correctly-specified approaches to MCMC for full Bayesian modeling under the elastic net are the 
Metropolis--Hastings (MH) approach of \citet{hans:11}, the exchange algorithm (EX) approach of \citet{wang:23}, 
and the transformation and rejection sampling (RS) approach introduced in this paper. Because we are not 
introducing new statistical
models, our comparisons focus on dynamics of the Markov chains generated by the MCMC algorithms.
We use effective sample size (ESS) for model parameters as a measure of efficacy of a given algorithm for a 
given data set. ESS is computed using the R package \texttt{mcmcse} \citep{mcmcse} based on \citet{gong:15}.

We consider the four simulation settings used by \citet{zou:05} in their original study of the elastic net
\citep[see also][]{hans:11}.
In Simulation~1, $n=20$ observations are simulated from the normal linear regression model
$y = X\beta + \varepsilon$ with $\beta = (3, 1.5, 0, 0, 2, 0, 0, 0)^T$. The error term is generated according
to $\varepsilon \sim \mbox{N}(0, \sigma^2 I_n)$ with $\sigma = 3$. Each row of the $n \times p $ design 
matrix $X$ is generated independently from a $\mbox{N}(0, V)$ distribution with covariance matrix $V$, 
where $V_{ij} = 0.5^{|i-j|}$ for $1 \leq i,j \leq p = 8$. Simulation~2 uses $\beta_j = 0.85$, $j = 1, \ldots, 8$, but
is otherwise the same as Simulation~1. Simulation~3 considers a higher-dimensional setting and larger
sample size, with $n=100$ and $p=40$.
The regression coefficients are $\beta_j = 2$ for $j = 1, \ldots, 10$ and $j = 21, \ldots, 30$, and $\beta_j = 0$
for all other $j$. The errors are generated with $\sigma = 15$, and the regressors are generated with 
$V_{ij} = 0.5$ for $i \neq j$ and $V_{ii} = 1$. Simulation~4 is the same as Simulation~3, but
$\beta_j = 3$ for $j = 1, \ldots, 15$, $\beta_j = 0$ for $j = 16, \ldots, 40$, and the regressors are
generated using a block-diagonal covariance matrix $V$ as follows. The first block ($1 \leq i,j \leq 5$)
has variances of $1.01$ and covariances of $1$, the second ($6 \leq i,j \leq 10$) and third 
($11 \leq i,j \leq 15$) are the same as the first, and the final block is the $25 \times 25$ identity
matrix. Covariances between all blocks are zero.

Fifty data sets $y$ and $X$ were generated for each simulation setting. No explicit ``intercept'' term was 
included in the generation of the response variables, but we follow the discussion in Section~\ref{sec:intro}
and construct posterior distributions based on the integrated likelihood (\ref{eq:intlik}) obtained by
marginalizing an intercept term from the regression model under a flat prior. Computationally, this means
using the mean-centered $y^*$ and $X^*$ data in all expressions reported in this paper.

We illustrate the performance of the MH and RS algorithms under the data augmentation representation
of the differentially-scaled prior for $\beta$ in (\ref{eq:correctDA1}) and (\ref{eq:correctDA2}). The prior
on $\sigma^2$ is the same in both cases, with $\nu_a = \nu_b = 1$, but we consider two different
priors for $\lambda_1$ and $\lambda_2$. The first prior has $L = \nu_1 = R = \nu_2 = 1$, which 
results in a uniform prior for $\lambda_1/(\lambda_1 + \lambda_2)$, the proportion of the total penalty
allocated to the $|\beta|_1$ term in the prior. We therefore call this prior the ``weak'' prior, as it does
not strongly favor either term in the penalty function. The second
prior has $L = 6$, $\nu_L = 4$, $R = 2$, and $\nu_R = 4$. This prior puts a $\mbox{Beta}(6,2)$ prior
on $\lambda_1/(\lambda_1 + \lambda_2)$ and represents a biasing of the prior in favor of stronger
$\ell_1$-norm penalization. We call this prior the ``strong'' prior.

The RS algorithm is implemented as described in Section~\ref{sec:sampdiff} under the data augmentation
representation. The MH algorithm updates $\log \sigma^2$, $\log \lambda_1$, and $\log \lambda_2$ 
using random-walk Metropolis updates as described in \citet{hans:11}.
Standard deviations for the random-walk innovations, $s_{\sigma^2}$, $s_{\lambda_1}$, and $s_{\lambda_2}$,
must be tuned and selected. Poor choices of these parameters can result in poorly-mixing chains. 
Experimentation with the simulated data settings revealed that $s_{\sigma^2} = 1$, $s_{\lambda_1}=1$, 
and $s_{\lambda_2} = 1$ resulted in generally good performance, and so we used these values in our
comparisons.

We compare the efficacy of the MH and RS MCMC algorithms with the exchange algorithm (EX) MCMC
sampler of \citet{wang:23}, which requires tuning only a single Metropolis--Hastings update as part of its sampling 
scheme. We note that the prior distributions used by \citet{wang:23} are different than those used
in this paper, and so the posterior distribution sampled by the EX algorithm is different than the posteriors
sampled by the RS and MH algorithms (which, for a given prior strength, are identical).
The prior for $\beta$ under the EX algorithm setup is parameterized using a differential scaling, and the penalty 
parameters are transformed via $(\lambda_1, \lambda_2) \rightarrow (\lambda = \lambda_1 + \sqrt{\lambda_2}, 
\alpha = \lambda_1/(\lambda_1 + \sqrt{\lambda_2}))$. The prior distributions for $\lambda$ and $\alpha$
do not contain any tunable parameters and so do not allow for the direct inclusion of prior information about the 
relative strengths of the two penalty perms. One step in the MCMC algorithm requires specification of $s_\alpha$, 
the standard deviation of a random-walk proposal for sampling the parameter $\log(\alpha/(1-\alpha))$. Poor
choices for this scale parameter can result in poorly-mixing Markov chains. We consider two values of
$s_\alpha$ in our simulations. The value $s_\alpha = 1$ was chosen because it resulted in generally
well-performing samplers across the four simulations; we refer to this algorithm as ``EX''. The value
$s_\alpha = 0.1$ was chosen because in general it resulted in poor-performing samplers; we refer to this
algorithm as ``EX-B''.

For each of the fifty simulated data sets in each of the five simulation setups, we compute the ESS 
for parameters $\beta_j$, $\sigma^2$, $\lambda_1$, and $\lambda_2$ based on 10,000
MCMC iterations of each algorithm after a burn-in of 100 iterations (with starting values chosen so that
convergence to stationarity should be very quick). We also compute the ESS for $\lambda$ and $\alpha$
(the transformed parameters under the EX algorithm) for completeness. We are comparing ESS across
six different sampler/posterior setups: the new rejection sampling methods under the weak and strong
prior setups (RS-W and RS-S, respectively), the random-walk Metropolis sampler under the weak and
strong prior setups (MH-W and MH-S, respectively), and the exchange algorithm sampler under the
good and poor step sizes (EX and EX-B, respectively).

To reduce nuisance variability when making comparisons, we treat the simulated data sets as a blocking
factor in the simulation experiment and compute, for each simulated data set, percentage improvement in 
ESS for MH-W, MH-S, RS-W, RS-S, and EX-B relative to the ESS for the EX algorithm,
$(\mbox{ESS}_{\mathrm{method}} - \mbox{ESS}_{\mathrm{EX}})/\mbox{ESS}_{\mathrm{EX}} \times 100$. 
Figure~\ref{fig:simbetas} displays these percentage improvements
(or reductions, if negative) for the $\beta_j$ across the five simulation studies. We discuss primarily
the comparison between the RS-* and EX-* methods because the MH-* methods often perform no
better on average than the RS-* methods and can require extensive step-size tuning to work well. Focusing first 
on Simulations~1 and 2 where $p = 8$, we see that the distribution of percentage improvement in ESS
for the RS-* methods is strongly right-skewed when the true $\beta_j \neq 0$. For some simulated data 
sets the RS-* methods perform worse than EX or EX-B, the reduction in ESS in those cases is small relative to the 
upside improvement when RS-* performs better. When the true $\beta_j = 0$, $j \in \{3, 4, 6, 7, 8\}$, the 
distribution of percent improvement for RSS-* is more symmetric
around zero, with the EX-* and RS-* methods performing similarly. This is quantified in Table~\ref{tab:avgimp},
where we see that the average improvement for RS-* is positive relative to EX for all regression coefficients
(with larger improvements within a simulation setting when $\beta_j \neq 0$).

In contrast to Simulations~1 and 2, the higher-dimensional ($p=40$) setting in Simulation~3 indicates
that the EX-* methods both tend to perform better than the RS-* methods. This simulation setting 
has two blocks of regressors with $\beta_j = 2$ ($j = 1, \ldots, 15$ and $j = 21, \ldots, 30$), and two
blocks of regressors with $\beta_j = 0$. All regressors are equally and moderately correlated via an
exchangeable covariance structure. We display ESS improvement for the first two $\beta_j$ in each of the four 
blocks. While the EX and EX-B methods perform similarly, the RS-* methods both perform worse
than they did in Simulations~1 and 2 relative to EX. While the percent improvement in ESS tends to be
negative for the RS-* methods in Simulation~3, we note that the reduction tends to be modest, with most
values falling between 0\% and -20\%. This is quantified in Table~\ref{tab:avgimp}, where we see that 
the average ESS percentage reduction tends to be modest (and maxes out at around 20\%) under both
priors across the simulated data sets.

Finally, Simulation~4 is another higher-dimensional setting ($p=40$) but where the regressors are
generated using a block-diagonal covariance matrix. ESS results are shown for the first two
$\beta_j$ in each of the four blocks; the true values are non-zero in the first three blocks and zero
in the last block. There is strong correlation among the regressors within blocks, but independence across 
blocks. We see that RS-S tends to performs as well or slightly better than EX for the nonzero coefficients,
with RS-W tending to perform slightly worse. All methods perform similarly when $\beta_j = 0$. The
average improvement results in Table~\ref{tab:avgimp} again indicate that when the RSS-* methods
perform worse, the reduction in ESS is not too large.

Figure~\ref{fig:simothers} displays the same information for the other parameters. The distributions of
the ESS improvements for the RS-* methods tend to be symmetric or right skewed across the simulations.
While in some cases the RS-* methods show a reduction in ESS relative to EX (e.g., $\sigma^2$ in 
Simulations~1 and 2), in others we see a much larger improvement (e.g., $\lambda_1$ for RS-S across
all simulations). The impact of the poorly-chosen step size for $\alpha$ in the exchange algorithm is
most apparent in this figure: the effective sample sizes for $\alpha$ tend to be much worse for EX-B
than for EX. In terms of average percent improvement, Table~\ref{tab:avgimp} shows that RS-S does
quite well for $\sigma^2$, $\lambda_1$, and $\lambda_2$ parameters in Simulations~1 and 2. Despite 
performing reasonably well for the the regression coefficients, $\beta_j$, RS-W does not do particularly 
well in any of the settings for $\sigma^2$, $\lambda_1$, and $\lambda_2$.

\begin{figure}[p]
 \begin{center}
 \includegraphics[scale=0.45]{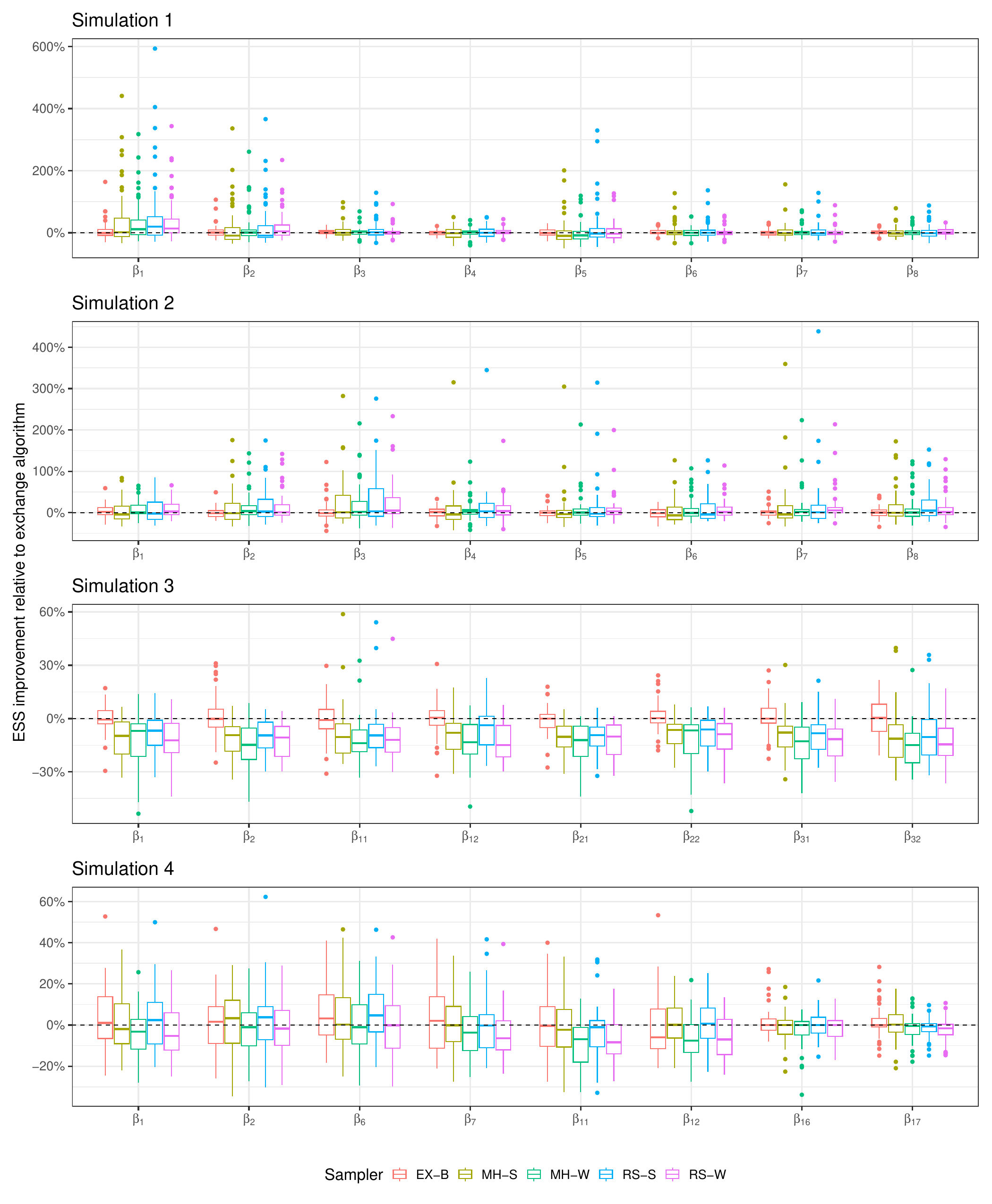}
 \caption{Percent improvement in ESS for fifty simulated data sets under four simulation settings for several 
 MCMC algorithms. The exchange algorithm (EX) with a well-chosen proposal standard deviation is the 
 baseline. ESSes for all eight $\beta_j$ are shown for Simulations~1 and 2; a selection of relevant $\beta_j$
 are shown for Simulations~3 and 4. EX-B: exchange algorithm with a poorly-chosen step size; MH-*:
 random-walk Metropolis updates for $\sigma^2$, $\lambda_1$, and $\lambda_2$ under Strong and
Weak priors; RS-*: the new rejection sampling methods under the two prior settings.
 	        \label{fig:simbetas}}
 \end{center}
\end{figure}

\begin{figure}[p]
 \begin{center}
 \includegraphics[scale=0.45]{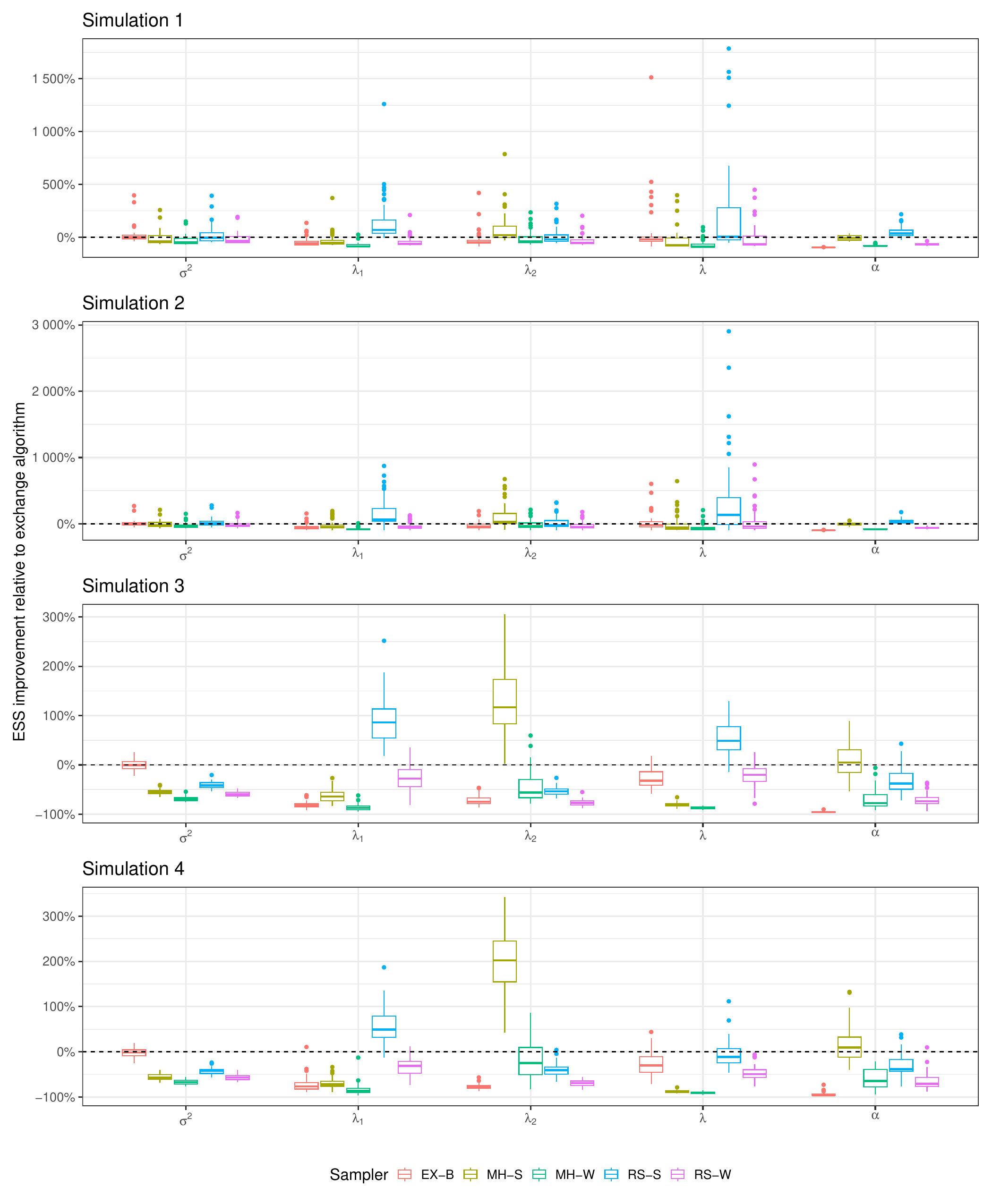}
 \caption{Percent improvement in ESS for fifty simulated data sets under four simulation settings for several 
 MCMC algorithms. The exchange algorithm (EX) with a well-chosen proposal standard deviation is the 
 baseline. The parameters $\lambda = \lambda_1 + \sqrt{\lambda_2}$ and $\alpha = \lambda_1 / 
 (\lambda_1 + \sqrt{\lambda_2})$ are a core part of the exchange algorithm sampler.
EX-B: exchange algorithm with a poorly-chosen step size; MH-*:
 random-walk Metropolis updates for $\sigma^2$, $\lambda_1$, and $\lambda_2$ under Strong and
Weak priors; RS-*: the new rejection sampling methods under the two prior settings.
 	        \label{fig:simothers}}
 \end{center}
\end{figure}

\begin{table}[!h]
\begin{center}
\begin{tabular}[h]{llrrrrrrrr} & Prior
 & $\beta_1$ & $\beta_2$ & $\beta_3$ & $\beta_4$ & $\beta_5$ & $\beta_6$ & $\beta_7$ & $\beta_8$ \\
\hline
\multirow{2}{*}{Simulation 1} & Weak & \textbf{42.85} & \textbf{19.98} & \textbf{3.19} & \textbf{2.24} & \textbf{5.03} & \textbf{1.58} & \textbf{1.57} & \textbf{2.79}\\
 & Strong & \textbf{59.73} & \textbf{22.63} & \textbf{8.64} & \textbf{2.69} & \textbf{18.06} & \textbf{7.26} & \textbf{4.92} & \textbf{2.48}\\
\hline
\multirow{2}{*}{Simulation 2} & Weak & \textbf{7.61} & \textbf{15.23} & \textbf{25.92} & \textbf{9.98} & \textbf{10.32} & \textbf{8.41} & \textbf{15.45} & \textbf{11.85}\\
 & Strong & \textbf{5.87} & \textbf{16.12} & \textbf{28.74} & \textbf{12.26} & \textbf{11.77} & \textbf{7.15} & \textbf{17.17} & \textbf{17.56}\\
\hline
\end{tabular}
\end{center}

\begin{center}
\begin{tabular}[t]{lrrrrrr}
%\hline
& Prior & $\sigma^2$ & $\lambda_1$ & $\lambda_2$ & $\lambda$ & $\alpha$ \\
\hline
\multirow{2}{*}{Simulation 1} & Weak & -14.12 & -34.15 & -28.88 & \textbf{36.99} & -60.1\\
%\cline{2-7}
 & Strong & \textbf{21.79} & \textbf{149.86} & \textbf{11.9} & \textbf{213.22} & \textbf{49}\\
\cline{1-7}
\multirow{2}{*}{Simulation 2} & Weak & -21.40 & -43.86 & -39.82 & \textbf{2.62} & -65.76\\
%\cline{2-7}
 & Strong & \textbf{19.83} & \textbf{166.75} & \textbf{18.39} & \textbf{353.03} & \textbf{38.27}\\
\cline{1-7}
\multirow{2}{*}{Simulation 3} & Weak & -58.98 & -25.86 & -76.51 & -21.78 & -71.12\\
%\cline{2-7}
 & Strong & -40.95 & \textbf{90.42} & -53.17 & \textbf{53.84} & -31.35\\
\cline{1-7}
\multirow{2}{*}{Simulation 4} & Weak & -56.86 & -33.72 & -69.57 & -47.39 & -65.9\\
%\cline{2-7}
 & Strong & -42.93 & \textbf{58.47} & -39.56 & -4.85 & -31.16\\
\hline
\end{tabular}
\end{center}
\caption{Average percent improvement in ESS across fifty simulated data sets for the MCMC algorithm using rejection sampling 
(RS) versus the exchange algorithm (EX) with a well-chosen proposal standard deviation. Positive numbers in bold indicate better
average performance for RS, e.g.,~$36.77$ indicates the ESS for RS was $36.77\%$ larger than it was for EX. 
\label{tab:avgimp}}
\end{table}

\newcolumntype{R}{>{\raggedleft\arraybackslash}X}
\newcolumntype{C}{>{\centering\arraybackslash}X}

\begin{table}
\begin{center}
%\begin{tabular}[h]{llcccc} 
%\begin{tabularx}{\textwidth}{ll*{4}{X}}
\begin{tabularx}{\textwidth}{lr*{4}{R}}  % for right-aligned
%\begin{tabularx}{\textwidth}{ll*{4}{C}}  % for centered
\multicolumn{2}{c}{} & \multicolumn{4}{c}{Percentage of the forty $\beta_j$ with average \% improvement in ESS:} \\
 &  & $\geq 0\%$ & $\geq -5\%$ & $\geq -10\%$ & $\geq -20\%$ \\ 
% \cline{3-6}
\hline
 \multirow{2}{*}{Simulation 3} & Weak Prior & 0 & 0 & 15 & 100 \\
 & Strong Prior & 0 & 2.5 & 97.5 & 100 \\
 \hline
 \multirow{2}{*}{Simulation 4} & Weak Prior & 10 & 92.5 & 100 & 100 \\
 & Strong Prior & 37.5 & 100 & 100 & 100 \\
 \hline
\end{tabularx}
\end{center}
\caption{Percentage of the number of $\beta_j$ (out of 40) for which the average percent improvement (or reduction,
when negative) in ESS for RS-* relative to EX is greater than or equal to various thresholds,
	\label{tab:40betas}}
\end{table}

\section{Discussion}\label{sec:disc}
Full Bayesian inference for the Bayesian elastic net regression model is challenging due to the $\Phi(\cdot)$
term in the normalizing constant for the prior on $\beta$ that is a function of $\lambda_1$, $\lambda_2$, and
sometimes $\sigma^2$ (depending on the parameterization of the model). All existing, correctly-specified 
methods for posterior sampling use at least one Metropolis--Hastings update that requires specification and
tuning of a proposal distribution. We have introduced transformations for the commonly- and differentially-scaled
priors (under both direct and DA representations) that result in ``well known'' full conditional distributions that can be 
easily sampled from for all but one parameter, $\theta$. Careful analysis of the full conditional for $\theta$ reveals
that rejection sampling approaches that take advantage of log-concavity of the target density function can be
used to efficiently produce samples directly from the full conditional. A key to this approach is that the rejection
sampling methods are automatic in the sense that no tuning is required. Access to MCMC algorithms that 
practitioners can run directly without having to interactively manipulate will make these statistical models
more broadly impactful.

We compared our new sampling methods with the existing, exchange algorithm-based approach of
\citet{wang:23}. The exchange algorithm sampler is cleverly designed to avoid computation of $\Phi(\cdot)$
while at the same time minimizing the number of tuning parameters, and it worked quite well across the simulation 
settings we considered. Despite its strong performance, there are several reasons why one might prefer the new 
rejection sampling approaches. First, the posterior distribution in \citet{wang:23} from which the
exchange algorithm is sampling has no user-specified parameters in the prior. While this might be useful from
a particular objective Bayes point of view where a goal is to provide users with methods that can be used reliably
without much user input, having the ability to shift the prior toward stronger $\ell_1$- or $\ell_2$-norm penalization
gives the user the ability to move beyond default prior settings to build a model that is a better match for their
analysis. Second, while the exchange algorithm was not too sensitive to the choice of the random-walk scale
parameter, $s_\alpha$, particularly poor choices reduced the ESSes for some of the parameters. Avoiding the
need to specify such a parameter is desirable.

We also note that the higher effective sample sizes for the parameters $\sigma^2$, $\lambda$, and $\alpha$
under the exchange algorithm-based MCMC sampler might be due to other sampling techniques used by
\citet{wang:23} to improve the Markov chain dynamics. \citet{wang:23} use a partially-collapsed 
\citep{liu:94, vand:08} update for $\sigma^2$ and implement a generalized Gibbs step \citep{liu:00, liu:04}
for $\sigma^2$, $\lambda$, and $\alpha$ that is based on a scale transformation group to accelerate 
convergence of the chain. Incorporating similar techniques into the rejection samplers introduced in this
paper may produce similar improvements and is an area of future work. 

Finally, we note that the performance of MCMC algorithms can be very sensitive to parameterization
The (re)parameterization introduced in this paper was chosen primarily to produce full conditional
distributions with log-concave densities that could be sampled easily via automatic rejection 
sampling techniques. Other transformations might be available that yield similar computational
simplicity while also further improving the Markov chain dynamics and will be the subject of future
research.

%\section*{Acknowledgments}
%...

\appendix

\section{Rejection sampling $\sigma^2$ with an inverse gamma proposal}\label{app:sig2rejsamp}
\citet{li:10} propose using rejection sampling to sample from the full conditional 
for $\sigma^2$ in their formulation of the Bayesian elastic net. 
They define the function $f(\sigma^2) \propto \pi(\sigma^2 \mid y, \beta, \tau^2, \lambda_1, \lambda_2)$ 
\citep[see equation~(6) in][]{li:10} to be
\begin{equation}
	f(\sigma^2) = (\sigma^2)^{-a-1} \left\{\Gamma_U\left(\frac{1}{2},
		\frac{\lambda_1^2}{8\sigma^2\lambda_2}\right)\right\}^{-p}
		e^{-b/\sigma^2}, \label{eq:litarget}
\end{equation}
where $a = n/2+p$ and $b > 0$ depends on $y$ and the other parameters in the model. We note that the
full conditional in \citet{li:10} is similar to but slightly different than the full conditional in (\ref{eq:dacssig2}) due 
to small differences in model specification and parameterization.

\citet{li:10} propose
rejection sampling from $f(\sigma)$ using an inverse gamma proposal distribution with density function
$h(\sigma^2) = b^a \Gamma(a)^{-1}(\sigma^2)^{-a-1}e^{-b/\sigma^2}$, where $a$ and $b$ are the same as 
in (\ref{eq:litarget}). \citet{li:10} describe
the following algorithm for rejection sampling using this proposal distribution: ``generate a candidate $Z$ from $h$
and a $u$ from uniform(0,1), and then accept $Z$ if $u \leq \Gamma\left(\frac{1}{2}\right)^p b^a f(Z)/\Gamma(a) h(Z)$
or, equivalently, if $\log(u) \leq p \log \left( \Gamma\left(\frac{1}{2}\right)\right) - p \log \left( \Gamma_U\left(\frac{1}{2},
\frac{\lambda_1^2}{8Z\lambda_2}\right)\right)$.'' By the definition of the upper incomplete gamma function, we have
\begin{eqnarray}
	\Gamma_U\left(\frac{1}{2}, \frac{\lambda_1^2}{8Z \lambda_2}\right) &=& 
		\int_{\lambda_1^2/(8Z\lambda_2)}^\infty t^{-1/2} e^{-t}dt \nonumber \\
		&\leq& \int_0^\infty t^{-1/2} e^{-t}dt \label{eq:reversal} \\
		&=& \Gamma\left(\frac{1}{2}\right), \nonumber
\end{eqnarray}
which means that $p \log \left( \Gamma\left(\frac{1}{2}\right)\right) - p \log \left( \Gamma_U\left(\frac{1}{2},
\frac{\lambda_1^2}{8Z\lambda_2}\right)\right) \geq 0$ with probability one. We also have that $\log(u) \leq 0$ with
probability one, and so the rejection sampling algorithm will accept every proposal with probability one, which
can only be true if $f(\sigma^2) \propto h(\sigma^2)$, which is not the case, and the algorithm does not produce
samples from $f(\sigma^2)$. The contradiction appears to arise from a reversal of the bound in (\ref{eq:reversal})
in one step of the construction of the algorithm in \citet{li:10}.
In fact, an inverse gamma distribution with parameters $a$ and $b$ that match those in (\ref{eq:litarget}) cannot
be used as a proposal distribution for rejection sampling. The ratio of the target to the proposal densities is
$f(\sigma^2)/h(\sigma^2) = b^{-a} \Gamma(a) \left\{\Gamma_U\left(\frac{1}{2},\frac{\lambda_1^2}{8\sigma^2\lambda}
\right)\right\}^{-p}$. We see that $\lim_{\sigma^2 \rightarrow 0} f(\sigma^2)/h(\sigma^2)= \infty$ and rejection sampling cannot
be implemented using this proposal distribution.

%%%%%%%%%%%%%%%%%%%%%%%%%%%%%
%% Bibliography
%%%%%%%%%%%%%%%%%%%%%%%%%%%%%

%\newpage
\bibliographystyle{./natbib}
\bibliography{./ensampling}

\vspace{1in}

%%%%%%%%%%%%%%%%%%%%%%%%%%%%%
%% Tables
%%%%%%%%%%%%%%%%%%%%%%%%%%%%%
%\renewcommand{\baselinestretch}{1}
%\begin{table}[m]
% \begin{center}

% \end{center}
%\end{table}

%\renewcommand{\baselinestretch}{1.85}

%%%%%%%%%%%%%%%%%%%%%%%%%%%%%
%% Figures
%%%%%%%%%%%%%%%%%%%%%%%%%%%%%

\end{document}